%% file: 0_main.tex
\documentclass[sigconf, nonacm]{acmart}
\usepackage{import}
\usepackage{subfig}
\usepackage{overpic}
\usepackage{balance}
\usepackage[a-2b]{pdfx}[2018/12/22]
\usepackage[ruled,vlined,linesnumbered]{algorithm2e} 

\newcommand\vldbdoi{10.14778/3603581.3603597}
\newcommand\vldbpages{2591 - 2604}
\newcommand\vldbvolume{16}
\newcommand\vldbissue{10}
\newcommand\vldbyear{2023}
\newcommand\vldbauthors{\authors}
\newcommand\vldbtitle{\shorttitle} 
\newcommand\vldbavailabilityurl{https://github.com/ymzgithub/ymzgithub-traj_LDP_submission_sup}
\newcommand\vldbpagestyle{empty}

\begin{document}
\title{Trajectory Data Collection with Local Differential Privacy}

\author{Yuemin Zhang}
\authornote{Work partially done at The Hong Kong Polytechnic University.}
\affiliation{%
  \institution{Harbin Engineering University}
}
\email{zhangyuemin@hrbeu.edu.cn}

\author{Qingqing Ye}

\affiliation{%
  \institution{Hong Kong Polytechnic University}
}
\email{qqing.ye@polyu.edu.hk}

\author{Rui Chen}
\authornote{Corresponding authors.}
\affiliation{%
  \institution{Harbin Engineering University}
}
\email{ruichen@hrbeu.edu.cn}

\author{Haibo Hu}
\affiliation{%
  \institution{Hong Kong Polytechnic University}
}
\email{haibo.hu@polyu.edu.hk}

\author{Qilong Han}\authornotemark[2]
\affiliation{%
  \institution{Harbin Engineering University}
}
\email{hanqilong@hrbeu.edu.cn}

\begin{abstract}
Trajectory data collection is a common task with many applications in our daily lives. Analyzing trajectory data enables service providers to enhance their services, which ultimately benefits users. However, directly collecting trajectory data may give rise to privacy-related issues that cannot be ignored. Local differential privacy (LDP), as the $\emph{de facto}$ privacy protection standard in a decentralized setting, enables users to perturb their trajectories locally and provides a provable privacy guarantee. Existing approaches to private trajectory data collection in a local setting typically use relaxed versions of LDP, which cannot provide a strict privacy guarantee, or require some external knowledge that is impractical to obtain and update in a timely manner. To tackle these problems, we propose a novel trajectory perturbation mechanism that relies solely on an underlying location set and satisfies pure $\epsilon$-LDP to provide a stringent privacy guarantee. In the proposed mechanism, each point's adjacent direction information in the trajectory is used in its perturbation process. Such information serves as an effective clue to connect neighboring points and can be used to restrict the possible region of a perturbed point in order to enhance utility. To the best of our knowledge, our study is the first to use direction information for trajectory perturbation under LDP. Furthermore, based on this mechanism, we present an anchor-based method that adaptively restricts the region of each perturbed trajectory, thereby significantly boosting performance without violating the privacy constraint. Extensive experiments on both real-world and synthetic datasets demonstrate the effectiveness of the proposed mechanisms.
\end{abstract}

\maketitle

\pagestyle{\vldbpagestyle}
\begingroup\small\noindent\raggedright\textbf{PVLDB Reference Format:}\\
\vldbauthors. \vldbtitle. PVLDB, \vldbvolume(\vldbissue): \vldbpages, \vldbyear.\\
\href{https://doi.org/\vldbdoi}{doi:\vldbdoi}
\endgroup
\begingroup
\renewcommand\thefootnote{}\footnote{\noindent
This work is licensed under the Creative Commons BY-NC-ND 4.0 International License. Visit \url{https://creativecommons.org/licenses/by-nc-nd/4.0/} to view a copy of this license. For any use beyond those covered by this license, obtain permission by emailing \href{mailto:info@vldb.org}{info@vldb.org}. Copyright is held by the owner/author(s). Publication rights licensed to the VLDB Endowment. \\
\raggedright Proceedings of the VLDB Endowment, Vol. \vldbvolume, No. \vldbissue\ %
ISSN 2150-8097. \\
\href{https://doi.org/\vldbdoi}{doi:\vldbdoi} \\
}\addtocounter{footnote}{-1}\endgroup

\ifdefempty{\vldbavailabilityurl}{}{
\vspace{.3cm}
\begingroup\small\noindent\raggedright\textbf{PVLDB Artifact Availability:}\\
The source code, data, and/or other artifacts have been made available at \url{\vldbavailabilityurl}.
\endgroup
}

\input{1_intro.tex}
\input{2_related_work.tex}

\input{3_preliminaries.tex}
\input{4_methodologies.tex}
\input{5_optimization.tex}
\input{6_experiments_conclusion.tex}

\begin{acks}
This work was supported by the National Key R\&D Program of China (Grant No. 2020YFB1710200), the National Natural Science Foundation of China (Grant No. 62072136, 62102334, 62072390, and 92270123), and the Research Grants Council, Hong Kong SAR, China (Grant No. 15218919, 15203120, 15226221, 15225921, 15209922, and C2004-21GF).
\end{acks}

\balance
\bibliographystyle{ACM-Reference-Format}
\bibliography{mybibliography}

\end{document}

%% file: 1_intro.tex
\section{Introduction}

Nowadays mobile applications generate a considerable amount of personal private data, such as locations and browsing histories. The collection and analysis of such data help data collectors improve their services, which benefits both data collectors and users. As a common and important type of personal data, trajectory data have a wide range of applications in data mining tasks. For example, taxi companies could help drivers decide where to wait for passengers by analyzing passengers' trajectory data~\cite{YZZ2011}. However, in the absence of proper privacy protection, direct collection of sensitive trajectory data would give rise to serious privacy concerns. Differential privacy (DP)~\cite{DMN2006} is regarded as a golden standard that provides a rigorous privacy protection. However, it requires a trusted server to which users must upload their raw data. This requirement  may not always be possible in practical settings. In contrast, local differential privacy (LDP)~\cite{KLN2011} enables distributed users to perturb their data locally and send the sanitized data to an untrusted server, which is more suitable for private trajectory data collection.

However, due to the known low utility of LDP perturbation on trajectory release, most of the existing approaches for private trajectory data collection either use some relaxed versions of LDP or require external knowledge that is difficult to obtain and update in practice. The relaxations of LDP, such as geo-indistinguishability~\cite{ABC2013}, cannot provide the strict privacy guarantee of $\epsilon$-LDP. The only approach that satisfies $\epsilon$-LDP is the NGRAM mechanism~\cite{CCF2021}, which uses additional external knowledge, such as business hours and category information, to boost utility. Such information may not be readily available in practice and is often not updated in a timely manner. \textcolor{black}{Therefore, our goal is to introduce a more practical trajectory perturbation mechanism that satisfies $\epsilon$-LDP while relying solely on the essential location set of trajectories.}

It is a non-trivial technical challenge to design such a mechanism with high utility. \textcolor{black}{The location points of a trajectory are usually the elements of a finite point set, e.g., grid centers or points of interest (POIs).} Each point has its corresponding tuple of latitude and longitude. To perturb a trajectory under $\epsilon$-LDP, we must ensure that each point in the set has a probability of being included in the output. However, the point domain is normally so large that the perturbed point may be far from the real one, which leads to poor utility. On the other hand, a real-world trajectory places a natural restriction that the distance between any two neighboring points cannot be too large. 
To resolve these problems, a straightforward solution is to restrict the possible perturbed points for each original point in the trajectory, as used in~\cite{CCF2021}. \textcolor{black}{However, such a solution requires not only the trajectory and the distribution of all possible locations, but also public external information, from which we would like to free ourselves.}

The key observation of this paper is that a user's trajectory is a time-ordered sequence of locations backed by the user's intentions, i.e., she needs to go to a few target places step by step. When she needs to go to a target place, she typically first determines an approximate direction from the current location to the target location.
 This direction information is the ``clue'' that connects every pair of adjacent points in a trajectory, which has not been explored in the existing approaches under LDP. \textcolor{black}{Motivated by this observation, in this paper, we propose a {\bf pivot sampling} perturbation mechanism that enhances privacy-aware trajectory modeling by sanitizing a given target point in a trajectory through a two-stage process.} \textcolor{black}{We first identify the adjacent points preceding and following a target point in the trajectory as its pivots.} 
 A pivot of the target point is treated as an origin so that the direction between itself and the target point can be obtained. \textcolor{black}{In the second stage, we perturb the target point within the region defined by bi-directional information from the two pivots. It is worth noting that our proposed mechanism is applicable to not only trajectories satisfying the above key observation, but also many other types, as long as the points in a trajectory belong to a finite location point set.}
 
In contrast to pivot sampling leveraging bi-directional information to confine a single point's perturbed result, it is also possible to confine the spatial region of an entire trajectory. According to the first law of geography~\cite{T1970}, neighboring points in a trajectory are usually close to each other, and therefore the spatial region of a trajectory is highly likely to be relatively small in comparison with the entire area of the map. Thus, we further introduce an anchor-based method to limit the spatial region of the entire trajectory, and then apply pivot sampling to trajectory perturbation. A simple strategy 
is to restrict trajectories to different spatial regions with the same fixed size. We theoretically analyze the utility of this strategy under different sizes. 
Then, to avoid the use of additional hyperparameters (i.e., the region size), we put forward an adaptive strategy to achieve better performance than the mechanism using the best fixed size.

 The key contributions of our study are summarized as follows:
\begin{itemize}
    \item[$\bullet$] We \textcolor{black}{introduce an original privacy-aware trajectory modeling solution to private trajectory data collection. We model} the perturbation of the points in a trajectory as a two-stage process and propose a novel mechanism called pivot sampling, which captures bi-directional information of the points to restrict their regions. We also propose a guideline on choosing the direction granularity. To the best of our knowledge, our study is the first to combine direction information with trajectory perturbation in LDP. 
    \item[$\bullet$] We also propose an anchor-based pivot sampling mechanism that restricts the spatial regions of trajectories while satisfying $\epsilon$-LDP. We perform a theoretical utility analysis of using a fixed region size, and then propose a strategy to adaptively restrict the spatial region of a trajectory.
    \item[$\bullet$] We perform extensive experiments on real-world and synthetic datasets to demonstrate the effectiveness of our mechanism in comparison with existing solutions.
\end{itemize}

The remainder of this paper is organized as follows. We discuss the related work in Section~\ref{related work}. In Section~\ref{preliminary}, we provide necessary background information and the problem definition. In Sections~\ref{method} and~\ref{Optimization}, we describe the proposed mechanisms in detail. In Section~\ref{experiment}, we present the experimental results. Finally, we conclude our paper in Section~\ref{conclusion}. 

%% file: 2_related_work.tex
\section{Related Work}\label{related work}
To prevent privacy leakages, several classic privacy \textcolor{black}{models} have been proposed, such as $k$-anonymity~\cite{S2001,S2002}, $l$-diversity~\cite{MKG2007}, and $t$-closeness~\cite{LLV2006}. These models make different assumptions about the background knowledge possessed by attackers. In recent years, DP~\cite{DMN2006} has been considered the golden standard for privacy protection. In contrast with other privacy models, DP does not need to assume background knowledge owned by attackers, and provides a rigorous privacy guarantee for individual user data. It has been widely used in many tasks such as \textcolor{black}{frequent subgraph mining~\cite{XSX2016,CSX2018}, high-dimensional data publishing~\cite{ZCP2017, CXZ2015, STC2016, CTS2020}, and sequential data sanitization~\cite{CFD2012}}. In the typical setting of DP, data owners upload their data to a trusted data collector, which then perturbs and shares the querying results. However, it comes with deficiencies in some practical settings---on the one hand, the DP model requires a trusted server to collect the raw data from users; on the other hand, users are unlikely to agree to directly upload their sensitive data to the collector in real life. To avoid these problems, LDP~\cite{KLN2011} was proposed. Many studies~\cite{YHM2019,YHA2020} have focused on applying LDP to protect users' local data, and several companies~\cite{CJK2018} have already developed LDP in ther products, such as Google~\cite{EPK2014}, Apple~\cite{A2017}, and Microsoft~\cite{DKY2017}. In the LDP model, data owners must perturb their sensitive data locally and then send the perturbed version to an untrusted data collector.

DP/LDP generally considers a trade-off between utility and privacy. Several relaxations of DP/LDP have been proposed to obtain better data utility for various data analysis tasks at the cost of weakened privacy protection. For example, $d_{\mathcal{X}}$-privacy~\cite{CAB2013} assumes that the indistinguishability of two inputs is inversely proportional to the distance between them, whereas in pure LDP the indistinguishability remains the same and does not depend on the distance. Geo-indistinguishability~\cite{ABC2013} is a variant of $d_{\mathcal{X}}$-privacy in the context of location privacy. Personalized LDP~\cite{CLQ2016} is a variant of LDP that enables different users to use different privacy levels determined by themselves.

Although many studies have investigated location privacy, few studies have examined private trajectory data collection in the local setting. Most studies have focused on trajectory privacy in the centralized setting~\cite{CFD2012,YCH2022,GLT2018,HCM2015,GRL2020,SYH2023}. Chen \textit{et al}.~\cite{CFD2012} utilize the prefix tree structure to group sequences with identical prefixes into the same branch. 
He \textit{et al}.~\cite{HCM2015} develop a mechanism for trajectory data synthesis, which discretizes raw trajectories using hierarchical reference systems to capture individual movements at different speeds. They also propose a post-processing strategy to restore directionality while sampling synthetic trajectories from the noisy model based on the Markovianity of the trajectory. \textcolor{black}{We would like to emphasize that, unlike our idea, their solution does not utilize the adjacent direction information to restrict the perturbation domain of a point. }\textcolor{black}{Furthermore, their solution is designed for DP and cannot be applied to LDP directly. }\textcolor{black}{Jiang \textit{et al.}~\cite{JSB2013} also consider utilizing direction information to publish trajectories under DP. They assume that the start and end points of a trajectory are known and that the distance between each pair of neighboring locations in a trajectory is bounded by a constant value. In contrast, we focus on LDP and aim to remove strong assumptions to make our solution more practical.} 

Most of the related studies in the local setting have been designed to satisfy the relaxations of LDP. To the best of our knowledge, the n-gram-based method NGRAM~\cite{CCF2021} is the only solution to private trajectory data collection that satisfies pure LDP. NGRAM makes use of public external knowledge, such as business hours and POI categories, for POI trajectory perturbation. However, it suffers from some limitations in real-world applications where the required public external information is often not readily available or cannot be updated in a timely manner. Therefore, in this paper we are motivated to remove the dependency on such public knowledge for better applicability.

%% file: 3_preliminaries.tex
\section{Preliminaries and Problem Formulation}\label{preliminary}
\subsection{Local Differential Privacy}
As a variant of differential privacy (DP)~\cite{DMN2006}, LDP~\cite{KLN2011} does not rely on a trusted data collector, and thus becomes a practical privacy notion for many applications in the distributed setting, such as frequency and mean estimation~\cite{XYH2022,WBL2017,DYH2022,WXY2019}, heavy hitter identification~\cite{BBC2021,QYY2016}, and preference ranking analysis~\cite{YCS2022}. A formal definition of LDP is given as follows.

\begin{definition}[$\epsilon$-LDP]
\label{def:ep_ldp}
Given a privacy budget $\epsilon>0$, a randomized mechanism $\mathcal{M}:\mathcal{X}\to\mathcal{Y}$ provides $\epsilon$-LDP if and only if, for any two inputs $x, x^{\prime}\in \mathcal{X}$ and any possible output $y\in \mathcal{Y}$, the following inequality holds
\begin{equation}\label{eq1}
\mathrm{Pr}[\mathcal{M}(x)=y]\leq e^{\epsilon} \times \mathrm{Pr}[\mathcal{M}(x^{\prime})=y].
\end{equation}
\end{definition}

LDP enables users to locally perturb their data in order to guarantee plausible deniability, which is controlled by privacy budget $\epsilon$. As with DP, LDP enjoys the desired properties of sequential composition and post-processing~\cite{DR2014}.
\begin{theorem}\label{theorem_sequential_composition}
    (Sequential Composition) Let each $\mathcal{M}_{i} (1\leq i \leq n)$ denote a mechanism satisfying $\epsilon_{i}$-LDP. Then the sequential combination of these mechanisms satisfies $\epsilon$-LDP, where $\epsilon=\sum_{i=1}^n\epsilon_{i}$.
\end{theorem}

\begin{theorem}
    (Post-Processing) Let $\mathcal{M}:\mathcal{X}\to\mathcal{Y}$ be a mechanism that satisfies $\epsilon$-LDP, and $f:\mathcal{Y}\to\mathcal{Y}^\prime$ be an arbitrary randomized mapping. Then \textcolor{black}{$f\circ\mathcal{M}:\mathcal{X}\to\mathcal{Y}^\prime$ satisfies $\epsilon$-LDP}, where $\circ$ denotes the composition of $f $ and $\mathcal{M}$.
\end{theorem}
From the sequential composition property, the amount of the privacy budget consumed will increase linearly when we apply the mechanisms sequentially; from the post-processing property, there is no additional privacy cost when perturbed data are used for further processing.

\subsection{Perturbation Mechanisms}

Some perturbation mechanisms that satisfy $\epsilon$-LDP have been proposed in the literature. In this subsection, we briefly introduce three classic mechanisms adopted in our study: $k$-ary randomized response ($k$-RR)~\cite{W1965, KOV2014}, exponential mechanism (EM)~\cite{MT2007} and square-wave mechanism (SW)~\cite{LWL2020}.

\begin{definition}[$k$-RR]
Given a privacy budget $\epsilon$, $k$-RR mechanism $\mathcal{M}:\mathcal{X}\to\mathcal{X}$ satisfies $\epsilon$-LDP, \textcolor{black}{if, for any input $x \in \mathcal{X}$, the output $y\in\mathcal{X}$} is sampled from the following distribution:
\begin{equation}\label{eq2}
\mathrm{Pr}[\mathcal{M}(x)=y]=
\begin{cases}
  \frac{e^{\epsilon}}{\mathcal{|X|}-1+e^{\epsilon}} & \text{ if } y=x, \\
  \frac{1}{\mathcal{|X|}-1+e^{\epsilon}} & \text{ otherwise}.
\end{cases}
\end{equation}
\end{definition}

By Equation~\ref{eq2}, $k$-RR always reports the original value (i.e., $y=x$) with a larger probability; otherwise, the mechanism returns one of the other values uniformly. It has been proven that $k$-RR achieves good performance when \textcolor{black}{the domain size $k$ (i.e., $\mathcal{|X|}$)} is not overly large (i.e., when $k<3e^\epsilon + 2$)~\cite{WBL2017}.

\begin{definition}[Exponential Mechanism]
    The exponential mechanism (EM) $\mathcal{M}$ preserves $\epsilon$-LDP, if, for any input $x \in \mathcal{X}$, the probability of any output $r \in \mathcal{R}$ is
    \begin{equation}\label{eq4}
    \mathrm{Pr}[\mathcal{M}(x)=r]=
        \frac{\exp(\frac{\epsilon u(x,r)}{2\Delta u})}{\sum_{r^{\prime}\in \mathcal{R}}\exp({\frac{\epsilon u(x,r^{\prime})}{2\Delta u}})},
    \end{equation}
    where $u(x,r)$ denotes the utility function of $x$ and $r$, and $\Delta u=\max_{x,x^\prime \in \mathcal{X}, \ r\in \mathcal{R}} |u(x,r)-u(x^\prime,r)|$ is the sensitivity of the utility score.
\end{definition}

The EM is a natural building block of answering queries with arbitrary utility scores, which makes it possible for a perturbed answer to be as close to the truth as possible (controlled by the utility score and privacy budget).

While both $k$-RR and the EM are used for categorical value perturbation, the Laplace mechanism~\cite{DMN2006} and Square-Wave Mechanism (SW) are typically designed for numerical values. Due to unbounded noise and high variance in the case of the Laplace mechanism, we alternatively adopt the SW mechanism~\cite{LWL2020} for numerical value perturbation.

\begin{definition}[Square-Wave Mechanism]
    The square-wave mechanism $\mathcal{M}$ satisfies $\epsilon$-LDP, if, for any input $x\in [0,1]$ and output $y\in [-b,b+1]$, the output probability of $y$ is:
    \begin{equation}\label{eq3}
        \mathrm{Pr}[\mathcal{M}(x)=y]=
        \begin{cases}
            \frac{e^{\epsilon}}{2be^{\epsilon}+1} & \text{if } |x-y|\leq b, \\
            \frac{1}{2be^{\epsilon}+1} & \text{otherwise},
        \end{cases}
    \end{equation}
    where $b=\frac{\epsilon e^{\epsilon}-e^{\epsilon}+1}{2e^{\epsilon}(e^{\epsilon}-1-\epsilon)}$.
\end{definition}
The SW mechanism takes a value in $[0,1]$ as input and returns a bounded output. The output range will have a higher concentration around the input value as the privacy budget increases, which is beneficial for utility enhancement.


\subsection{Problem Statement}
We consider a trajectory as a time-ordered sequence of points. Each user $u$ in the population possesses a sensitive trajectory $\tau=\{p_{1}, p_2, \cdots, p_{|\tau|}\}$, where $p_i \in \mathcal{P}$ denotes a location in the form of a two-dimensional (2D) point with latitude and longitude, and $\mathcal{P}$ is a finite point set. \textcolor{black}{We assume that the location point set $\mathcal{P}$ is public and accessible to the data collector and all users. This is a common requirement for location/trajectory privacy preservation~\cite{ABC2013,BCP2014}, and $\mathcal{P}$ is readily available in practice.} \textcolor{black}{Given two points, we use the Haversine distance as the distance measure, which calculates their great-circle distance, namely, the shortest distance over the surface of earth.} 

For the sake of privacy, each user locally perturbs her trajectory $\tau$ into $\hat{\tau}=\{\hat{p}_{1}, \hat{p}_2, \cdots, \hat{p}_{|\tau|}\}$ while satisfying $\epsilon$-LDP. \textcolor{black}{Following Definition~\ref{def:ep_ldp}, we aim to provide \emph{user-level} $\epsilon$-LDP for each user by perturbing \emph{every} point in her trajectory. Hence, the perturbed trajectory ensures that an adversary cannot infer \emph{any} point (i.e., location) she visited in the original trajectory with high confidence.}

Upon receiving the perturbed trajectories from all users, some trajectory data analysis tasks (e.g., preservation range queries) can be performed. The goal of our study is to preserve as much trajectory information as possible, while achieving a rigorous $\epsilon$-LDP privacy guarantee.

%% file: 4_methodologies.tex
\section{Methodology}\label{method}
In this section, we first introduce our design rationale for private trajectory data collection with LDP and provide an overview of the workflow in Section~\ref{subsec:rationale}. Then, we present the implementation details in Sections~\ref{subsec:implementation} and ~\ref{subsec:choose_granularity}. Finally, we establish the privacy guarantee for the proposed mechanism in Section~\ref{subsec:privacy}.

\subsection{Design Rationale and Overview}
\label{subsec:rationale}
In \textcolor{black}{the task of} private trajectory data collection, the distance between any two points (i.e., locations) is the most intuitive information that can be used for perturbation, e.g., applying EM~\cite{MT2007} to perturb each point~\cite{YLP2017,TCA2019}. However, as the domain of \textcolor{black}{each point in} a trajectory is a large region, the perturbed point is likely to be far from the original point. To address this challenge, a natural idea is to restrict the domain of a perturbed point. A good source of restriction is the direction between two neighboring points---for a given point, once the direction of the next point is determined, the size of the perturbation domain will be reduced significantly. Therefore, we leverage both distance and direction information for trajectory perturbation.

An overview of the proposed mechanism is shown in Figure~\ref{fig_1}. Each point in the trajectory is perturbed in a two-stage process. First, the direction between \textcolor{black}{the target point (i.e., the point to be perturbed)} and its neighbor is perturbed, and then this point is perturbed under the restriction of the perturbed direction. To this end, the original trajectory is duplicated to create two copies, from which pivot points (i.e., the pink ones) are selected alternately to ensure that each point in the trajectory has private access to its adjacent direction information. Then, the direction between each pair of neighboring points in the two trajectories is perturbed. Under the perturbed direction constraints, each non-pivot point is perturbed using a considerably smaller point domain. Finally, the optimal perturbed result is determined based on the two independently perturbed trajectories.

\subsection{Implementation of Proposed Mechanism}
\label{subsec:implementation}
In this subsection, we propose the pivo\underline{t} sam\underline{p}ling (TP) mechanism as the implementation of the above perturbation process. It exploits a novel strategy called pivot perturbation to capture additional direction information. An illustration  of the TP mechanism is shown in Figure~\ref{fig_2}a. The process consists of three steps: independent perturbation, pivot perturbation, and the optimal perturbed trajectory determination. To utilize direction information to restrict each point's spatial region while avoiding privacy leakage, we need to identify the pivots of each point. A pivot of a target point is treated as an origin so that the direction between the target point and itself can be calculated. The neighboring points of each point in the original trajectory are suitable candidates for pivots. However, such points are sensitive and, therefore, cannot be directly used. Furthermore, directly perturbing points in the original trajectory in a sequential manner would accumulate excessive noise. We duplicate the original trajectory to create two copies, and approximately half of the points in each trajectory can be selected as pivots in an alternate manner (i.e., by alternating between the two copies). Thus, the other \textcolor{black}{non-pivot} points that are not selected can access perturbed direction(s) between the neighboring pivot(s) and themselves to reduce their perturbation domains. Then, we alternately select points as the pivots in the two copies to obtain two perturbed trajectories and generate the optimal result based on them. In this way, each point in the final perturbed trajectory can enjoy both direction and distance information privately.

\begin{figure}[t]
\centering
    \includegraphics[width=0.96\columnwidth]{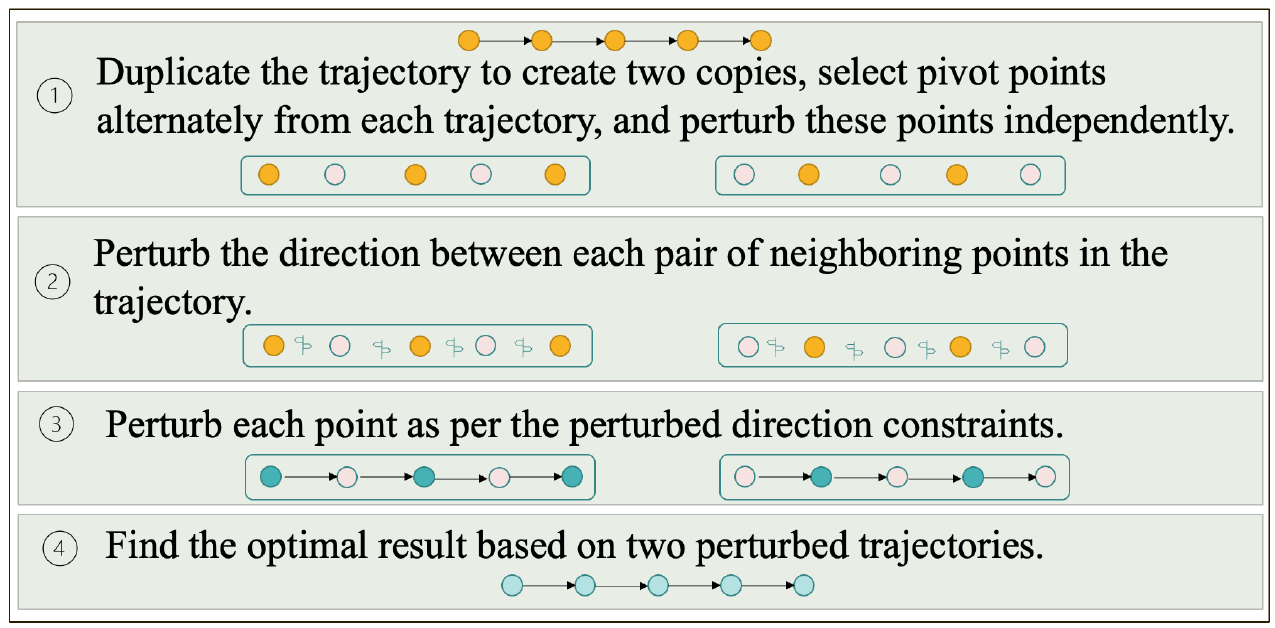}
    \caption{An overview of the proposed method (best viewed in color).}
    \label{fig_1}
\end{figure}

Let us take trajectory $\tau=\{p_{i}\in \mathcal{P}|1\leq i\leq |\tau|, |\tau|=5\}$ in Figure~\ref{fig_2}a as an example. The perturbation process is similar if $|\tau|$ is even. First, the trajectory $\tau$ is duplicated as $\tau^{\prime}$=$\tau^{\ast}$=$\tau$. \textcolor{black}{Then, the points in $\tau^\prime$ (resp. $\tau^\ast$) are divided into two sets --- a set of pivot points to be perturbed, and a set of non-pivot points getting access to the direction information.} The pivot points are selected alternately from $\tau^\prime$ (resp. $\tau^\ast$), each of which will be perturbed independently. Then the perturbation result is used to perturb non-pivot points in $\tau^\prime$ (resp. $\tau^\ast$), by utilizing the directional information between the target points and perturbed neighboring pivots, enabling each target point to enjoy a smaller perturbation domain.

Specifically, for $\tau^{\prime}$, we first select a set of pivot points (e.g., $p_2$ and $p_4$ in Figure~\ref{fig_2}a) denoted as $\mathcal{P}^{\tau^{\prime}}_{ind}=\{p_2,p_4\}$. These points are perturbed using EM:
\begin{equation}\label{eq5}
    Pr[\Hat{p}=r]=
        \frac{exp(\frac{\epsilon u(p,r)}{2\Delta u})}{\sum_{r^{\prime}\in \mathcal{P}}exp({\frac{\epsilon u(p,r^{\prime})}{2\Delta u}})},
\end{equation}
where $p\in \mathcal{P}^{\tau^{\prime}}_{ind}$, $r\in \mathcal{P}$, $u(p,r)=-dist(p,r)$, $dist(\cdot)$ is the Haversine distance, and the sensitivity $\Delta u=\max_{p,r,r^{\prime}\in \mathcal{P}} |u(p,r)-u(p,r^{\prime})|$. For instance, the pink perturbed points $\hat{p}^\prime_2$ and $\hat{p}^\prime_4$ in Figure~\ref{fig_2}a are the pivots in $\tau^\prime$, denoted as $\hat{\mathcal{P}}^{\tau^{\prime}}_{ind}=\{\hat{p}^\prime_2, \hat{p}^\prime_4\}$. Then, we can determine the perturbed directions between each point in $\mathcal{P}^{\tau^{\prime}}-\hat{\mathcal{P}}^{\tau^{\prime}}_{ind}=\{p_1, p_3, p_5\}$ and its neighboring pivots, i.e., the previous and the next perturbed neighboring points in $\tau^\prime$. 
\begin{figure}[t]
    \centering
    \subfloat[An illustration of the TP mechanism.]{
    \includegraphics[width=0.9\columnwidth]
    {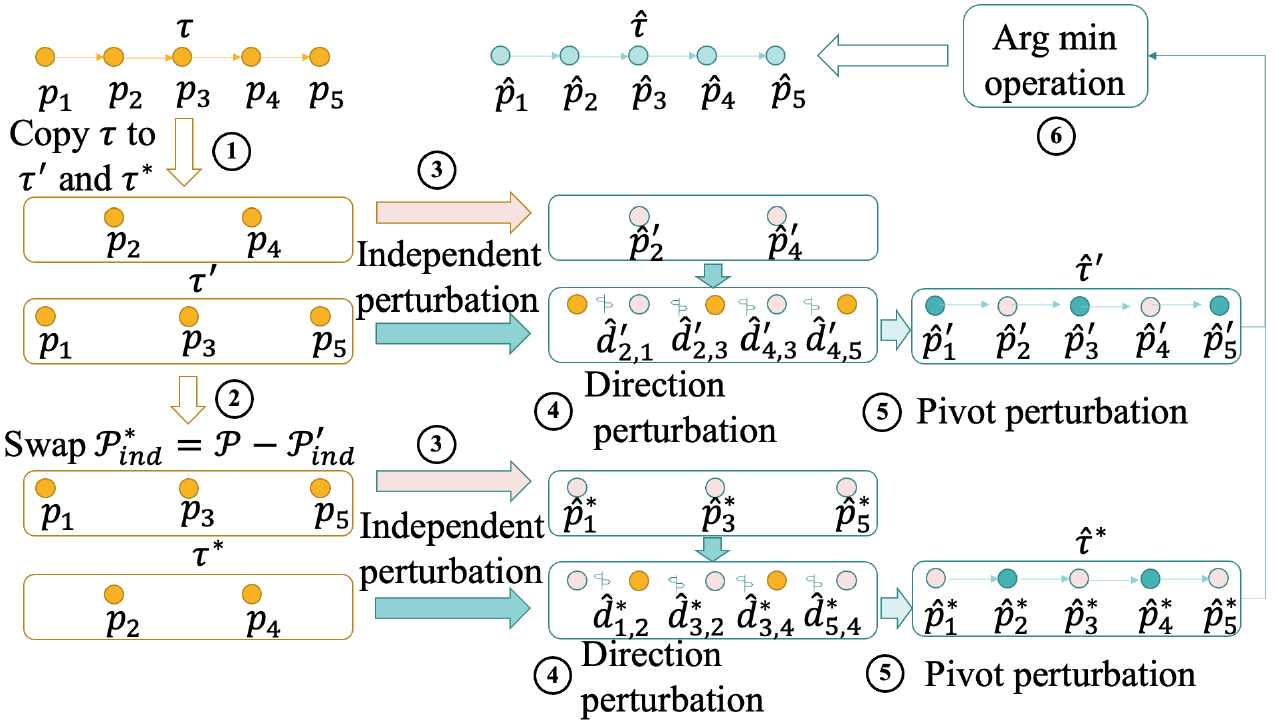}
    }\hfill
    \subfloat[Determine the perturbation domain of $p_3$.]{
    \includegraphics[width=0.4\columnwidth]
    {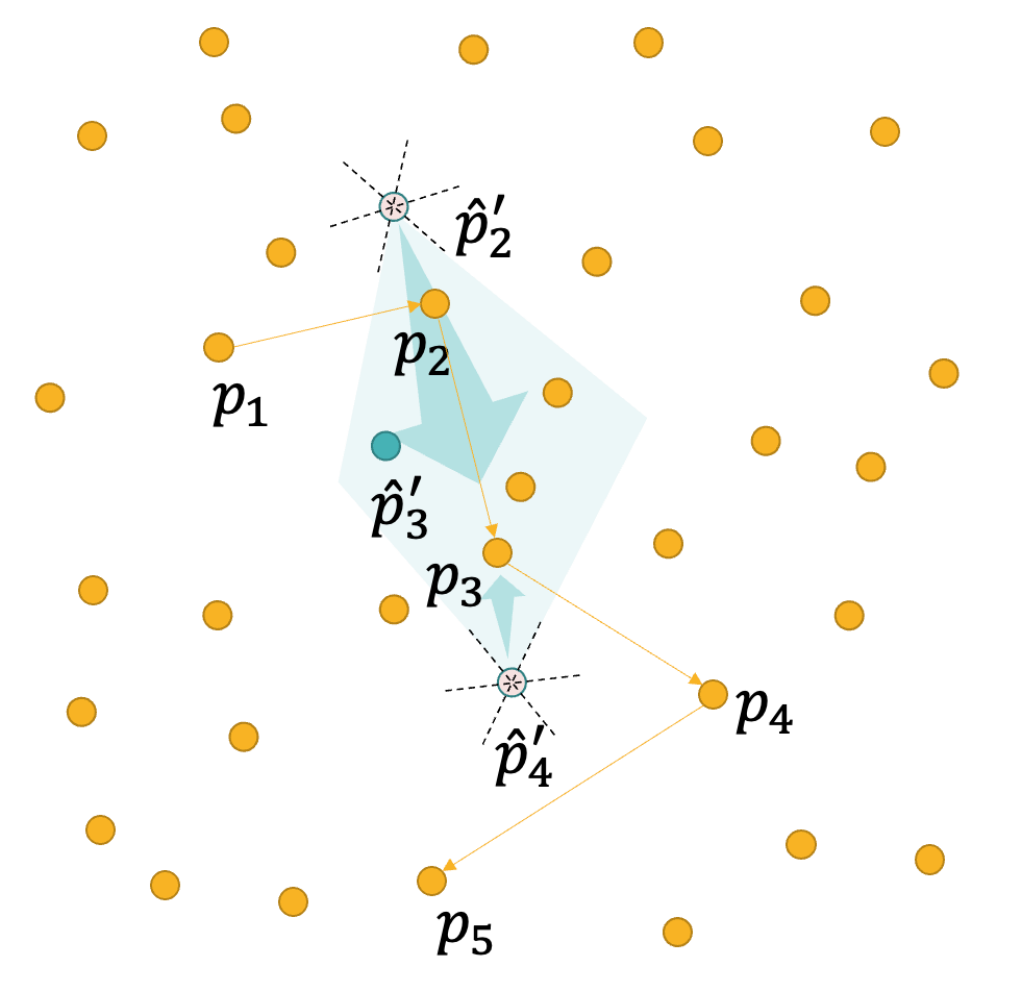}
    }\hfill
    \subfloat[Restrict trajectory region.]{
    \includegraphics[width=0.49\columnwidth]
    {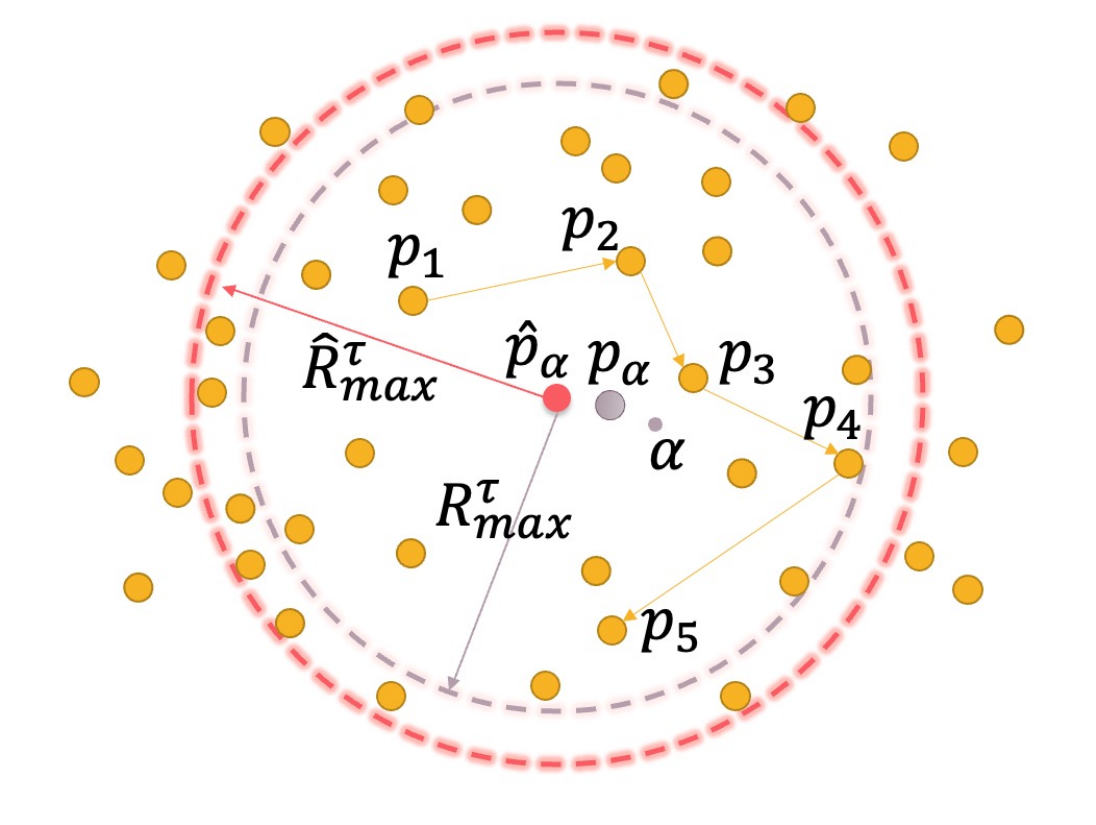}
    }
    \caption{Illustrations of the (A)TP mechanism.}
    \label{fig_2}
\end{figure}

To obtain a radian interval that covers the original direction with a high probability, randomized response~\cite{W1965, KOV2014} is adopted to perturb the directions. We must discretize the directions first, whose optimal granularity will be determined in Section~\ref{subsec:choose_granularity}. We take a neighboring pivot as the origin and the direction of the target point \textcolor{black}{(i.e., the non-pivot point to be perturbed under the direction constraint)} as the reference for the original discrete direction. 
Then the radian interval of a discrete direction with index $d\in \mathcal{D}$ is $[(2d-1)\frac{\pi}{|\mathcal{D}|},(2d+1)\frac{\pi}{|\mathcal{D}|}]$, where $\mathcal{D}=\{0,1,2, \cdots, |\mathcal{D}|-1\}$ denotes the discrete direction universe under a given granularity. \textcolor{black}{As shown in Figure~\ref{fig_2}b, when $|\mathcal{D}|=6$ and $p_3$ is the target point, $\hat{p}_2$ and $\hat{p}_4$ are designated as the origins for discretizing the directions, and the direction between $\hat{p}_2$ or $\hat{p}_4$ and $p_3$ serves as the median value of the central angle of the corresponding original discrete direction.} Let $\mathcal{D}^{\tau^\prime}$ denote the discrete directions to be perturbed in $\tau^\prime$. Then, in Figure~\ref{fig_2}a, we need to perturb the directions between each point in $\hat{\mathcal{P}}^{\tau^{\prime}}_{ind}$ and its perturbed neighboring pivots in $\tau^\prime$, namely, $\mathcal{D}^{\tau^{\prime}}=\{d_{2,1},d_{2,3},d_{4,3}, d_{4,5}\}$, where $d_{i,j}$ represents the index of the discrete direction from the $i$-th perturbed point to the $j$-th point in trajectory $\tau^\prime$. For instance, $d_{2,3}$ denotes the index of the discrete direction, the median value of whose central angle is towards the target point $p_3$ from the origin $\hat{p}^\prime_2$. 
The discrete direction is perturbed by $k$-RR:
\begin{equation}\label{eq6}
Pr[\mathcal{K}(d)=d^{\prime}]=
\begin{cases}
  \frac{e^{\epsilon}}{\mathcal{|D|}-1+e^{\epsilon}} & \text{ if } d=d^{\prime}, \\
  \frac{1}{\mathcal{|D|}-1+e^{\epsilon}} & \text{ if } d\neq d^{\prime}.
\end{cases}
\end{equation}

Let $\hat{\mathcal{D}}^{\tau^{\prime}}$ be the perturbed directions in $\tau^\prime$. Based on these perturbed directions and independently perturbed pivots, the perturbation domain of each remaining point in $\tau^\prime$ can be determined by Algorithm~\ref{alg1}. Note that the function $GetPointSet(\hat{p}_{j},\hat{d}_{j,k})$ in Line 3 returns the points located in the input radian interval $\hat{d}_{j,k}\in \hat{\mathcal{D}}^{\tau}$ of the point $\hat{p}_{j}\in \hat{\mathcal{P}}^{\tau}_{ind}$. 
For instance, in Figure~\ref{fig_2}b, $p_2$ and $p_4$ are perturbed independently to $\hat{p}^{\prime}_{2}$ and $\hat{p}^{\prime}_{4}$ (i.e., two pink points) as two neighboring pivots of the target point $p_3$. The green arrows represent the perturbed discrete directions from the two pivots as origins (i.e., $\hat{p}^{\prime}_{2}$ and $\hat{p}^{\prime}_{4}$) to the target point $p_3$ (assuming that two correct discrete directions are reported). Then the set of points located in the corresponding central angle are selected as the outputs, i.e.,  $GetPointSet(\hat{p}^{\prime}_{2}, \hat{d}_{2,3})$ and $GetPointSet(\hat{p}^{\prime}_{4}, \hat{d}_{4,3})$. Finally, the intersection of the two point sets is obtained as the output of Algorithm~\ref{alg1}, i.e., the perturbation domain of $p_3$, which consists of the points located in the green area shown in Figure~\ref{fig_2}b. 
By applying Algorithm~\ref{alg1}, the perturbation domains of $p_{1}$, $p_{3}$ and $p_{5}$ can be obtained as $\mathcal{P}^{1}$, $\mathcal{P}^{3}$ and $\mathcal{P}^{5}$, respectively. These smaller perturbation domains will be increasingly accurate as $\epsilon$ increases, thus facilitating the selection of better perturbed points. Next, the other points in $\tau^\prime$ can be perturbed with the corresponding perturbation domains using EM, that is, for each $p_{j}\in \mathcal{P}- \mathcal{P}^{\tau^{\prime}}_{ind}$, $\mathcal{P}$ is substituted with $\mathcal{P}^{j}$ in Equation~\ref{eq5} to obtain the perturbed $\Hat{p}^\prime_{j}$.

\begin{algorithm}[t]
  \small 
  \caption{GetPointDomain}\label{alg1}
  \KwIn{Target point $p_{i}$, point set $\mathcal{P}$, trajectory $\tau$, independently perturbed points $\hat{\mathcal{P}}^{\tau}_{ind}$ and perturbed directions $\hat{\mathcal{D}}^{\tau}$}
  \KwOut{The perturbation domain $\mathcal{P}^{i}$ of the target point}
  $\mathcal{P}^{i}=\emptyset$\;\label{alg1line1}
  \uIf{$1 < i < |\tau|$}
  {
    $\mathcal{P}^{i}\!=\! GetPointSet(\hat{p}_{i-1}, \hat{d}_{i-1,i}) \! \cap \! GetPointSet(\hat{p}_{i+1}, \hat{d}_{i+1,i})$\;\label{alg1line3} 
    \If{$\mathcal{P}^{i}=\emptyset$}
    {
      $\mathcal{P}^{i}=\mathcal{P}$\;\label{alg1line5}
    }
    $\mathcal{P}^{i}=\{p_{i}\}\cup \hat{\mathcal{P}}^{i}$\;\label{alg1line6}
  }\uElseIf{$i=1$}{
    $\mathcal{P}^{i}=\{p_{i}\}\cup GetPointSet(\hat{p}_{i+1}, \hat{d}_{i+1,i})$\;\label{alg1line8}
  }\Else{
    $\mathcal{P}^{i}=\{p_{i}\}\cup GetPointSet(\hat{p}_{i-1}, \hat{d}_{i-1,i})$\;\label{alg1line11}
  }
  \Return  $\mathcal{P}^{i}$\;
\end{algorithm}


The other copy $\tau^{\ast}$ can be processed in a similar way. 
This time, $\mathcal{P}^{\tau^{\prime}}_{ind}$ and $\mathcal{P}-\mathcal{P}^{\tau^{\prime}}_{ind}$ are swapped, i.e., for $\tau^{\ast}$, we set $\mathcal{P}^{\tau^{\ast}}_{ind}=\mathcal{P}^{\tau^{\prime}}-\mathcal{P}^{\tau^{\prime}}_{ind}$ and $\mathcal{P}^{\tau^{\ast}}-\mathcal{P}^{\tau^{\ast}}_{ind}=\mathcal{P}^{\tau^{\prime}}_{ind}$. For example, in Figure~\ref{fig_2}a, $\mathcal{P}^{\tau^{\ast}}_{ind}=\{p_{1},p_{3},p_{5}\}$ is perturbed into $\hat{\mathcal{P}}^{\tau^{\ast}}_{ind}=\{\hat{p}^\ast_1,\hat{p}^\ast_3,\hat{p}^\ast_5\}$ by Equation~\ref{eq5}. Then,  $\mathcal{D}^{\tau^{\ast}}=\{d_{1,2},d_{3,2},d_{3,4},d_{5,4}\}$ is perturbed to obtain $\hat{\mathcal{D}}^{\tau^{\ast}}=\{\hat{d}_{1,2},\hat{d}_{3,2},\hat{d}_{3,4},\hat{d}_{5,4}\}$. Based on these perturbed discrete directions, the other non-pivot points in $\tau^{\ast}$ are perturbed by EM to obtain $\hat{\tau}^{\ast}$.

By applying independent perturbation and pivot perturbation independently on both $\tau^\prime$ and $\tau^\ast$ (as shown in Algorithm~\ref{alg3}, which invokes Algorithm~\ref{alg2} twice with flag $\mathcal{F}=1$ and 0, respectively), two perturbed points are obtained for each point in $\tau$: one based on only the distance factor of the original point, and the other perturbed under the direction restriction of neighboring pivot(s). Then, these two trajectories are combined to obtain the optimal perturbed trajectory by solving the following equation:
\begin{equation}\label{eq8}
    \underset{\hat{\tau}}{\arg\min}\sum\nolimits_{i=1}^{|\hat{\tau}|}dist(\hat{p}_{i},\hat{p}^{\prime}_{i})+dist(\hat{p}_{i},\hat{p}^{\ast}_{i}),
\end{equation}

\begin{algorithm}[t]
	\small 
  \caption{Independent and Pivot Perturbation (Pivot)}\label{alg2}
  \KwIn{Trajectory $\tau$, point domain $\mathcal{P}$, privacy budget $\epsilon$, discrete direction set $\mathcal{D}$, flag $\mathcal{F}$}
  \KwOut{Perturbed trajectory $\hat{\tau}$}
Initialize $\hat{\mathcal{P}}^{\tau}_{ind}, \hat{\mathcal{D}}^{\tau}, \mathcal{P}_{res}$ and $\hat{\mathcal{P}}_{res}$ to $\emptyset$\;
  $\epsilon_d=0.75\epsilon$,
  $\epsilon_{ind}=(\epsilon-\epsilon_d)/2$,
  $\epsilon_{rest}=(\epsilon-\epsilon_d)/2$\;
  \For{$1\leq i\leq|\tau|$}{
  \tcp{Assuming $|\tau|$ is odd}
  \If{i\%2==$\mathcal{F}$}{
  Use Equation~\ref{eq5} with $\epsilon_{ind}/|\tau|$ to perturb the $i$-th point in $\tau$ and obtain $\hat{p}_{i}$\;
  $\hat{\mathcal{P}}^{\tau}_{ind}=\hat{\mathcal{P}}^{\tau}_{ind}\cup\{\hat{p}_i\}$\;
  }\Else{
  $\mathcal{P}_{res}=\mathcal{P}_{res}\cup\{p_i\}$\;
  }
  }
  \For{$p_j$ in $\mathcal{P}_{res}$}{
  Use Equation~\ref{eq6} with equal budget $\epsilon_{d}/2(|\tau|-1)$ to perturb the discrete directions $d_{i-1,i}$ and $d_{i+1,i}$\;
  $\hat{\mathcal{D}}^{\tau}=\hat{\mathcal{D}}^{\tau}\cup\{\hat{d}_{i,i-1},\hat{d}_{i,i+1}\}$\;
  $\mathcal{P}^{j}=GetPointDomain(p_j,\mathcal{P},\tau, \hat{\mathcal{P}}^{\tau}_{ind}, \hat{\mathcal{D}}^{\tau})$\;
  Use Equation~\ref{eq5} with $\epsilon_{rest}/|\tau|$ by substituting $\mathcal{P}$ with $\mathcal{P}^{j}$ to obtain $\hat{p}_j$\;
  $\hat{\mathcal{P}}_{res}=\hat{\mathcal{P}}_{res}\cup\{\hat{p}_j\}$\;
  }
  $\hat{\tau}=OrderedByOriginalIdx(\hat{\mathcal{P}}^{\tau}_{ind}\cup\hat{\mathcal{P}}_{res})$\;
  \Return  $\hat{\tau}$\;
\end{algorithm}

\noindent where $\hat{p}_{i}\in \mathcal{P}$ is the $i$-th point in $\hat{\tau}$, $\hat{p}^{\prime}_{i}\in \hat{\tau}^{\prime}$, $\hat{p}^{\ast}_{i}\in \hat{\tau}^{\ast}$, and $dist(\cdot)$ denotes the Haversine distance. Then, the final perturbed trajectory benefits from bi-directional information. However, the challenge remains on how to determine the optimal direction granularity.

\subsection{Choosing the Direction Granularity}
\label{subsec:choose_granularity}
From the perspective of the utility of the proposed mechanism, a coarse-grained direction would cover a greater number of points, leading to a higher recall rate for the chosen points. Furthermore, a smaller $|\mathcal{D}|$ increases the accuracy of the $k$-RR mechanism. Unfortunately, the coverage of a greater number of noisy points may result in a considerably large perturbation domain when EM is used to perturb points under the perturbed direction constraint. In contrast, a fine-grained direction can lead to a smaller perturbation domain, but faces more challenges to select the exact direction due to the large $|\mathcal{D}|$. In this case, even though it is difficult to identify the original discrete direction, there is chances to obtain discrete directions close to the original one. As $\epsilon$ increases, the probability of obtaining such a close discrete direction increases.

In a nutshell, a proper choice of the direction granularity represents a critical trade-off between the preservation of direction and the perturbation domain size of EM. The influence of the perturbation domain size depends not only on the chosen direction granularity but also on the distribution of points on the map. The analysis of the preservation probability of the target point to be perturbed is difficult because the distribution varies for each point. Therefore, we opt to design a general strategy to guide the choice of the direction granularity independent of a specific dataset. We focus on only the direction-preserving performance under different granularities based on the $k$-RR mechanism and different $\epsilon$ values.

\begin{algorithm}[t]
	\small 
  \caption{TP Mechanism (TP)}\label{alg3}
  \KwIn{Trajectory $\tau$, point set $\mathcal{P}$, privacy budget $\epsilon$}
  \KwOut{Perturbed trajectory $\hat{\tau}$}
  $\tau^{\prime}=\tau$,
  $\tau^{\ast}=\tau$\;
  Solve Equation~\ref{eqprob} to obtain the discrete direction set $\mathcal{D}$\;
  \tcp{Pivot perturbation, assuming $|\tau|$ is odd}
  $\hat{\tau}^{\ast} = Pivot(\tau^{\ast},\mathcal{P},\epsilon/2,\mathcal{D},\mathcal{F}=1)$\;
  $\hat{\tau}^{\prime} = Pivot(\tau^{\prime},\mathcal{P},\epsilon/2,\mathcal{D},\mathcal{F}=0)$\;
  \tcp{Find the optimal perturbed trajectory}
  Solve Equation~\ref{eq8} to obtain $\hat{\tau}$ based on $\hat{\tau}^{\ast}$ and $\hat{\tau}^{\prime}$\;
  \Return  $\hat{\tau}$\;
\end{algorithm}
\begin{figure}[t]
    \centering
    \subfloat[$\theta=\pi /12$]{
    \includegraphics[width=0.15\columnwidth]{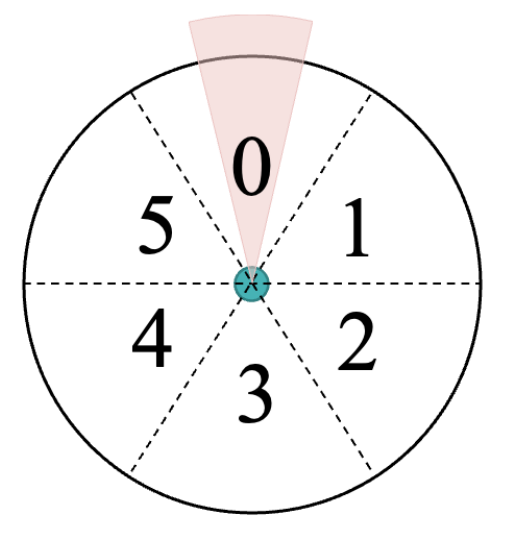}
    }\hfill
    \subfloat[$\theta=\pi /6$]{
    \includegraphics[width=0.17\columnwidth]{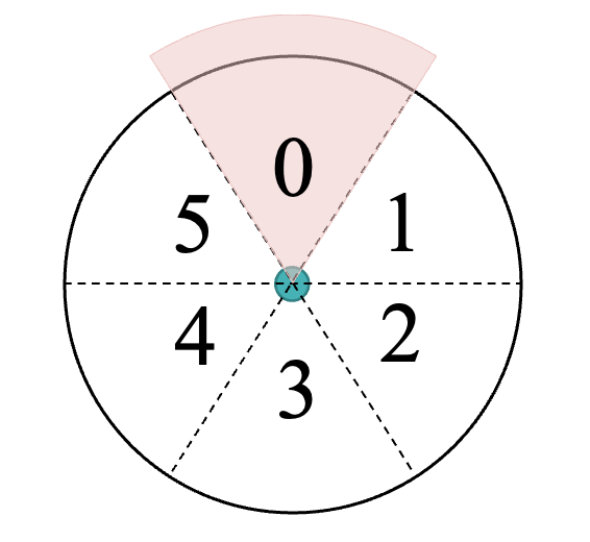}
    }\hfill
    \subfloat[$\theta=\pi /4$]{
    \includegraphics[width=0.15\columnwidth]{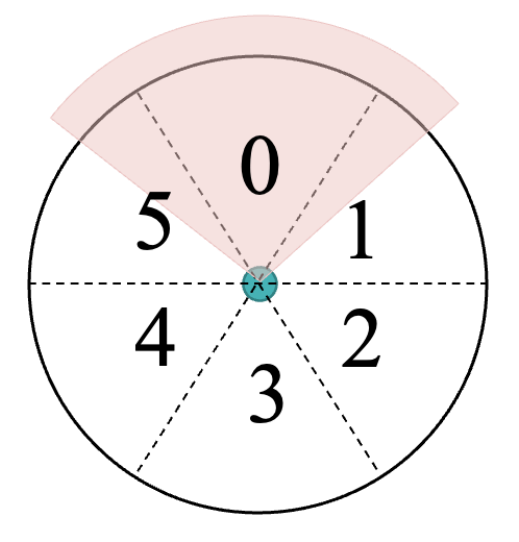}
    }\hfill
    \subfloat[$\theta=\pi /2$]{
    \includegraphics[width=0.17\columnwidth]{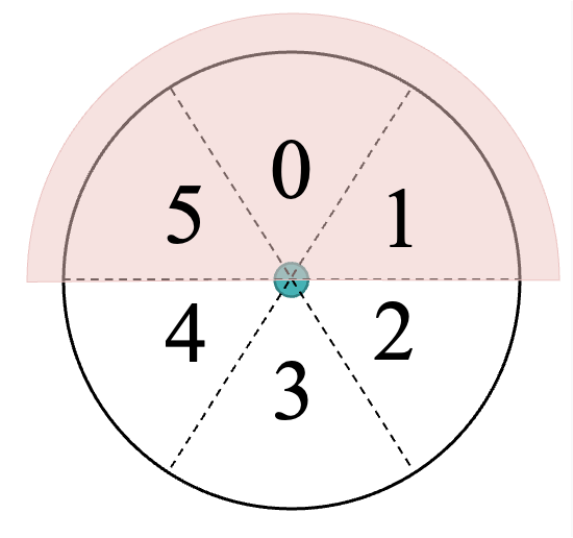}
    }
    \caption{Calculation of the \textcolor{black}{average direction-preserving success probability. The pink area denotes the query range $\theta$
    , with sector 0 serving as the input for $k$-RR, representing the original discrete direction, the median value of whose central angle is towards a target point from the green origin.
    }}
    \label{fig_4}
\end{figure}

Since the $k$-RR mechanism always returns the original discrete direction with the highest probability, evaluating the direction-preserving performance within just one constant range may not be the fairest approach. The granularity with the highest probability is always the one that is closest to the constant range. 
Therefore, we choose to compare the average direction-preserving success probability across different radian intervals of different granularities under a given $\epsilon$ value. To this end, we choose the granularity $|\mathcal{D}|$ as follows:
\textcolor{black}{\begin{equation}\label{eqprob}
     \underset{g \in \mathcal{D}_c}{\arg\max}\sum\nolimits_{\theta_j\in \Theta}\sum\nolimits_{d_i\in \mathcal{R}(g)}
     \varphi(d_i; \theta_j)\cdot \lambda(d_i; \epsilon; g)
     /|\Theta|,
\end{equation}
\begin{equation}\notag
     \varphi(d_i; \theta_j)=\frac{|[(2d_i-1)\frac{\pi}{g},(2d_i+1)\frac{\pi}{g}] \cap [-\theta_j, \theta_j]|}{2\pi /g},
\end{equation}
\begin{equation}\notag
    \lambda(d_i; \epsilon; g) =
\begin{cases}
  \frac{e^{\epsilon}}{g-1+e^{\epsilon}} & \text{ if } d_i=d, \\
  \frac{1}{g-1+e^{\epsilon}} & \text{ otherwise},
\end{cases}
\end{equation}
}
where \textcolor{black}{$\mathcal{D}_c$ denotes the candidate direction granularities, $\mathcal{R}(g)=\{0,1,2, \cdots, g-1\}$ denotes the discrete directions based on a candidate granularity $g$, $\Theta=\{\theta_j|\theta_j=\pi /g, g\in \mathcal{D}_c\}$ denotes the set of direction-preserving query ranges, and $d$ represents the correct discrete direction. 
The calculation of the average success probability can be treated as a weighted average of those results from direction-preserving range queries, i.e., the probabilities of preserving the correct radian intervals under different $g \in \mathcal{D}_c$ values when perturbed by $k$-RR with a given $\epsilon$ value.
} 

\textcolor{black}{For example, in Figure~\ref{fig_4}, assume that the current granularity $g=6 \in \mathcal{D}_c=\{2,4,6,12\}$, 0 is the original discrete direction, $\Theta=\{\pi/2,\pi/4,\pi/6,\pi/12\}$, and $q=\frac{e^{\epsilon}}{g-1+e^{\epsilon}}$. The query values in $\Theta$ correspond to candidate granularities. When the query is $\theta=\pi /12$, as shown in Figure~\ref{fig_4}a, $\lambda(d_i=0;\epsilon;g=6)$ returns $q$, and the weight $\varphi(d_i=0;\theta=\pi /12)$ is the proportion of the pink radian interval to the entire radian interval of the sector that is denoted as 0, i.e., $\frac{|[-\frac{\pi}{12},\frac{\pi}{12}] \cap [-\frac{\pi}{12}, \frac{\pi}{12}]|}{2\pi /6}=\frac{\pi/12-(-\pi/12)}{2\pi/6}=\frac{1}{2}$. The success probability is calculated in a similar way when the query is changed, as shown in Figure~\ref{fig_4}b-\ref{fig_4}d.}
The selection of $g$ is independent of a specific point, but relies on the privacy budget and candidate direction granularities. In particular, $g$ should be between $2$ and $12$; otherwise, the region under the perturbed direction constraint would be too small to cover the correct points. 
When $\epsilon$ increases, the perturbed direction is more accurate, and the benefits of the fine-grained direction becomes more significant. Through extensive experiments under different $\epsilon$ values, we observe that the granularity chosen by comparing the average direction-preserving success probability is close to the granularity with the best performance.

\begin{algorithm}[t]
	\small 
  \caption{ATP Mechanism}\label{alg5}
  \KwIn{Trajectory $\tau$, point set $\mathcal{P}$, privacy budget $\epsilon$, discrete direction set $\mathcal{D}$}
  \KwOut{Perturbed trajectory $\hat{\tau}$}
  \tcp{Copy the trajectory, assuming $|\tau|$ is odd}
  $\tau^{\prime}=\tau$,
  $\tau^{\ast}=\tau$\;
  \tcp{Independent and pivot perturbation}
  $\epsilon_R=0.25\epsilon$, 
  $\epsilon_3=\epsilon-\epsilon_R$\;
  $\hat{\mathcal{P}}^{\ast}_{\alpha}=RTR(\tau^{\ast},\mathcal{P},\epsilon_R/2$)\;\label{alg5line5}
$\hat{\mathcal{P}}^{\prime}_{\alpha}=RTR(\tau^{\prime},\mathcal{P},\epsilon_R/2$)\;\label{alg5line6}
  $\hat{\tau}^{\ast} = Pivot(\tau^{\ast},\hat{\mathcal{P}}^{\ast}_{\alpha},\epsilon_3/2,\mathcal{D},\mathcal{F}=1)$\;\label{alg5line7}
  $\hat{\tau}^{\prime} = Pivot(\tau^{\prime},\hat{\mathcal{P}}^{\prime}_{\alpha},\epsilon_3/2,\mathcal{D},\mathcal{F}=0)$\;\label{alg5line8}
  \tcp{Find the optimal perturbed trajectory}
  Solve Equation~\ref{eq8} to obtain $\hat{\tau}$ based on $\hat{\tau}^{\ast}$ and $\hat{\tau}^{\prime}$\;
  \Return  $\hat{\tau}$\;
\end{algorithm}

\subsection{Privacy Analysis}
\label{subsec:privacy}
In this subsection, we show that the proposed TP mechanism preserves differential privacy.
\begin{theorem}\label{theorem2}
    The TP mechanism satisfies $\epsilon$-LDP.
\end{theorem}
\begin{proof}
        For trajectory $\tau$ (we assume that the number of points $|\tau|$ in $\tau$ is odd. The proof is similar when it is even), raw sensitive data are accessed in three ways after copying $\tau$ into $\tau^{\prime}$ and $\tau^{\ast}$: independent point perturbations, direction perturbations, and other point perturbations based on the directions.

        For the independent point perturbations, we use EM with privacy budget $\epsilon_{ind}/|\tau|$ for each point. The points with an even index in $\tau$ are selected as pivots when perturbing $\tau^\prime$ and independently perturbed. Each independent perturbation of these points satisfies $\epsilon_{ind}/|\tau|-LDP$. The other points are selected as pivots when perturbing $\tau^{\ast}$ and independently perturbed. Each independent perturbation also satisfies $\epsilon_{ind}/|\tau|$-LDP. The total number of independent perturbations is $|\tau|$. 

        In the case of direction perturbations in pivot perturbations, recall that the direction is perturbed by considering the perturbed adjacent point(s) of a target point as the origin. By the post-processing property of LDP, making use of such an origin does not consume any additional privacy budget. Each perturbation for direction $d_{i,j}$ by $k$-RR satisfies $\epsilon_{d}/2(|\tau|-1)$-LDP. The total number of direction perturbations is $2(|\tau|-1)$.

        \textcolor{black}{Finally, based on the perturbed directions and the independently perturbed pivot points, each of the other non-pivot points in $\tau^{\prime}$ and $\tau^{\ast}$ is perturbed by EM with budget $\epsilon_{rest}/|\tau|$ (i.e., point perturbations under the direction restrictions in pivot perturbations). Then, $\hat{\tau}^{\prime}$ and $\hat{\tau}^{\ast}$ are obtained, and the optimal perturbed trajectory is determined. By the post-processing property, this step does not require any additional privacy budget.}
        Thus, according to the sequential composition theorem, the entire TP mechanism satisfies $\epsilon$-LDP, where $\epsilon=\epsilon_{ind}+\epsilon_{rest}+\epsilon_{d}$.
\end{proof}

%% file: 5_optimization.tex
\section{Anchor-based Pivot Sampling Mechanism}
\label{Optimization}
The TP mechanism uses pivot perturbation to capture bi-directional information and combines the direction information with independent perturbation of each point in a trajectory. Pivot perturbations can restrict the spatial region of the points in the trajectory to reduce the negative impact of the large point domain to a certain extent. However, a large number of independent perturbations still depend on a large $|\mathcal{P}|$. These independent perturbations introduce considerable amount of noise and almost uniformly select points in the large $\mathcal{P}$ when the privacy budget $\epsilon$ is small, which makes the perturbed points far from the actual region of the trajectory. To mitigate this problem, we propose an \underline{a}nchor-based pivo\underline{t} sam\underline{p}ling (ATP) mechanism in this section.

\begin{algorithm}[t]
	\small 
  \caption{Restrict Trajectory Region (RTR)}
  \label{alg4}
  \KwIn{Trajectory $\tau$, point set $\mathcal{P}$, privacy budget $\epsilon_R$}
  \KwOut{Perturbed point domain $\hat{\mathcal{P}}_{\alpha}$}
  Use Equation~\ref{eq10} to calculate the trajectory anchor $\alpha$ and map it to the nearest point $p_{\alpha}$\;\label{alg4line1}
  Use Equation~\ref{eq5} with $\epsilon_1=0.25\epsilon_R$ to perturb $p_{\alpha}$ into $\hat{p}_{\alpha}$\;\label{alg4line2}
  Get $R^\tau_{max}$ and perturb it using the SW mechanism with $\epsilon_2=\epsilon_R-\epsilon_1$\;\label{alg4line3}
  Calibrate the perturbed $\hat{R}^{\tau}_{max}$ to obtain $\hat{R}$ using Equation~\ref{eq27}\;\label{alg4line4}
  Get the perturbed point set $\hat{\mathcal{P}}_{\alpha}$ of $\tau$ based on $\hat{p}_{\alpha}$ and $\hat{R}$\;\label{alg4line5}

  \Return  $\hat{\mathcal{P}}_{\alpha}$\;
\end{algorithm}

\subsection{Restricting Trajectory Region}
According to the first law of geography~\cite{T1970}, the spatial region of a trajectory is likely to be relatively small when compared with the entire area of the map. The spatial region can be used to limit the domain of the points in a trajectory, which is used by our ATP mechanism to restrict the point domain of a trajectory. In contrast to the TP mechanism, the ATP mechanism (as shown in Algorithm ~\ref{alg5}) first restricts the trajectory region, i.e., the perturbation domain of a trajectory. Specifically, before perturbing each copy of the trajectory, the trajectory anchor is calculated and perturbed, and then the region size is determined based on the perturbed longest distance between the perturbed anchor and the trajectory (Lines ~\ref{alg5line5} and ~\ref{alg5line6}). After that, the trajectory can be perturbed in a sequential manner, similar to that in the TP mechanism (Lines ~\ref{alg5line7} and ~\ref{alg5line8}).

With regard to the trajectory region, ATP identifies the points that are more centralized and as close to the original region as possible while consuming an acceptable amount of the privacy budget. In what follows, we present a simple and lightweight method to determine this region. We first need to determine a way of representing a trajectory and its region. This representation must encode the information about the points in the trajectory, especially the geographic information, and the shape of the region needs to be carefully chosen to avoid privacy leakage and enhance data utility. As such, we propose an anchor-based method (see Algorithm~\ref{alg4}) that restricts each trajectory to a constant shape, i.e., a circular region. For each trajectory, the anchor $\alpha$ is first calculated and mapped to the nearest point as the anchor point $p_{\alpha}$ (Line~\ref{alg4line1}) as follows:
\begin{equation}\notag
    lat(\alpha)=\sum\nolimits_{i=1}^{|\tau|}lat(p_{i})/|\tau|\text{, }lon(\alpha)=\sum\nolimits_{i=1}^{|\tau|}lon(p_{i})/|\tau|,
\end{equation}
\begin{equation}\label{eq10}
    p_{\alpha}=\underset{p}{\arg\min}\text{ }dist(\alpha, p),
\end{equation}
where $p\in \mathcal{P}$, $\tau$ is the trajectory, $lat(p_i)$ and $lon(p_i)$ denote the latitude and longitude of $p_i$, respectively, and $dist(\cdot)$ is the Haversine distance. Then, the anchor point is perturbed using EM to obtain the perturbed point $\hat{p}_{\alpha}$ (Line~\ref{alg4line2}):
\begin{equation}\label{eq11}\notag
    Pr[\Hat{p}_{\alpha}=p]=
        \frac{exp(\frac{\epsilon u(p_{\alpha},p)}{2\Delta u})}{\sum_{p^{\prime}\in \mathcal{P}}exp({\frac{\epsilon u(p_{\alpha},p^{\prime})}{2\Delta u}})},
\end{equation}
where $p\in \mathcal{P}$, and $u(p_{\alpha},p)=-dist(p_{\alpha},p)$. The anchor $\alpha$ is mapped to $p_{\alpha}\in \mathcal{P}$ so that each possible $u(p_{\alpha},p)$ can be pre-calculated to increase the efficiency of the perturbation mechanism. For example, in Figure~\ref{fig_2}c, the anchor $\alpha$ is calculated and mapped to $p_{\alpha}$, and then $p_{\alpha}$ is perturbed to $\hat{p}_{\alpha}$ by EM. 

After the perturbed anchor point is obtained, the region of the trajectory can be determined. A simple way is to set a pre-defined and fixed value for the radius $R$. Then, the perturbation domain is obtained as follows:
\begin{equation}\label{eq12}\notag
    \hat{\mathcal{P}}_{\alpha}=\{p\in \mathcal{P}|dist(\hat{p}_{\alpha},p)\leq R\}.
\end{equation}
Since the spatial region depends on the radius $R$, it should be carefully chosen.

Intuitively, different users are likely to produce trajectories with varying area sizes. If a fixed radius $R$ is set, each trajectory will be restricted to the same granularity, resulting in loss of a considerable amount of information or introducing many noisy points in the candidate point set. To address this challenge, we propose an adaptive strategy to set different $R$ values for different trajectories.

We aim to determine $R$ values that are not too large and ensure that the region covers the points of a trajectory to the best extent possible while preserving privacy. The longest distance $R^{\tau}_{max}$ (i.e., the purple arrow in Figure~\ref{fig_2}c) between a point in the trajectory $\tau$ and the perturbed anchor point $\hat{p}_{\alpha}$ can be obtained as follows:
\begin{equation}\label{eq19}\notag
    R^{\tau}_{max}=\underset{p\in \mathcal{P}^{\tau}}{\max}\text{ }dist(\hat{p}_{\alpha}, p),
\end{equation}
where $\mathcal{P}^{\tau}$ denotes the set of the points in trajectory $\tau$. To ensure privacy, $R^{\tau}_{max}$ must be perturbed (Line~\ref{alg4line3} of Algorithm~\ref{alg4}). Since $R^{\tau}_{max}$ is the distance between the perturbed point $\hat{p}_{\alpha}\in \mathcal{P}$ and a point in the trajectory, i.e., $p\in \mathcal{P}^{\tau}\subseteq\mathcal{P}$, it is bounded by the maximum distance between $\hat{p}_{\alpha}$ and a point in $\mathcal{P}$. Therefore, we use SW mechanism, which is designed to perturb bounded continuous data. In comparison with the piecewise mechanism~\cite{LWL2020}, the SW mechanism is more concentrated, which is suitable for our task. As the input and output ranges of the SW mechanism are $[0,1]$ and $[-b,b+1]$, respectively, the distance is first normalized via $r=\frac{R^{\tau}_{max}}{\Delta R}$,
where \textcolor{black}{$\Delta R=\underset{p\in \mathcal{P}}{\max}\text{ }dist(\hat{p}_{\alpha}, p)$}. Then, $r$ is perturbed to $\hat{r}$ by Equation~\ref{eq3}. 
Finally, the output is mapped to the real distance range via $\hat{R}^{\tau}_{max}=\frac{(\hat{r}+b)\Delta R}{2b+1}$.


After the perturbed radius $\hat{R}^{\tau}_{max}$ (i.e., the red arrow in Figure~\ref{fig_2}c) is obtained, the regions of different trajectories will vary adaptively. 
Due to the property of the SW mechanism, the perturbed $\hat{R}^{\tau}_{max}$ is likely to be too small or large when given a small $\epsilon$. Here we make a few observations. First, the distribution of the region size of human mobility trajectory can be approximated by a power-law distribution~\cite{GHB2008}, implying that there should be only a few large sizes. On the other hand, the perturbed radius may become smaller due to the large randomness from a small $\epsilon$. In the case of a trajectory region, although a larger perturbed radius is likely to cover more points and leads to a larger perturbation domain, it can increase the coverage rate of the points in the original trajectory while being still relatively small in comparison with the entire area size. That is, the benefit of a larger radius outweighs the disadvantage of the noise it introduces when the original radius is small. To this end, we guide the perturbation by a calibration that ``pulls'' the perturbed radius close to a ``center'' when $\epsilon$ is small.

\subsection{Calibrating the Radius}
The radius calibration includes two steps: the center of calibration is first determined, and then the moving step size is calculated and the radius is calibrated based on $\epsilon$.

Recall that $R^{\tau}_{max}$ is bounded and is determined by the farthest point from $\hat{p}_\alpha$ in the trajectory. To identify a proper value $\eta$ for the calibration center based on $\epsilon$ and $\hat{R}^\tau_{max}$, we aim to amplify the impact of the points satisfying a possible range of the farthest point from $\hat{p}_\alpha$. Since the input domain of SW is $[0,1]$, we sample some test values uniformly to avoid privacy leakage and calculate the possible ranges. Let $\mathcal{V}=\{v_{k}|v_{k}=0.1\times k, k=0,1,...,10\}$ denote the test value set, $l_{k}=v_{k}-b$, and $u_{k}=v_{k}+b$, we have
\begin{small}
\begin{equation}\label{eq23}\notag
    \mathcal{V}^{R}=\{v_k\in \mathcal{V}|l_{k}\leq \frac{(2b+1)\hat{R}^{\tau}_{max}}{\Delta R}-b\leq u_{k}, 0\leq k\leq 10\},
\end{equation}
\end{small}
\begin{equation}\label{eq24}\notag
    l=\underset{v_k\in \mathcal{V}^{R}}{\min}\text{ }v_k,
    u=\underset{v_k\in \mathcal{V}^{R}}{\max}\text{ }v_k,
\end{equation}

\begin{equation}\label{eq26}\notag
    \mathcal{P}^{R}=\{p\in \mathcal{P}|l\leq \frac{(2b+1)dist(\hat{p}_{\alpha},p)}{\Delta R}-b\leq u\}.
\end{equation}
The values in $\mathcal{V}$ are used for testing, and for each value $v_{k}$, the output range with high probability $q=\frac{e^{\epsilon}}{2be^{\epsilon}+1}$ is used to determine whether it covers $\hat{R}^{\tau}_{max}$. The candidate points, i.e., the points whose distance to $\hat{p}_{\alpha}$ are within the desired range, are obtained. Then, the weighted average distance is calculated as the calibration center that assigns larger weights to the candidate points as follows:
\begin{equation}\label{eq30}\notag
    \eta=\frac{\sum_{p\in \mathcal{P}^{R}}{q\cdot dist(\hat{p}_{\alpha},p)}+\sum_{p^{\prime}\in \mathcal{P}-\mathcal{P}^{R}}{(1-q)dist(\hat{p}_{\alpha},p^{\prime})}}{q|\mathcal{P}^{R}|+(1-q)|\mathcal{P}-\mathcal{P}^{R}|}.
\end{equation}

After obtaining the calibration center $\eta$, we must determine how to ``pull'' the perturbed radius $\hat{R}^{\tau}_{max}$ close to it. As per the property of the SW mechanism, the perturbed value becomes more concentrated as the privacy budget $\epsilon$ increases. Therefore, the calibration effect should be reduced with increasing $\epsilon$. The perturbed radius $\hat{R}^{\tau}_{max}$ is calibrated as follows (Line~\ref{alg4line4} of Algorithm~\ref{alg4}):
\begin{equation}\label{eq27}
    \hat{R}=\hat{R}^{\tau}_{max} + \xi e^{-\epsilon},
\end{equation}
\begin{equation}\label{eq28}
    \xi=(\eta-\hat{R}^{\tau}_{max})\cdot \frac{1}{1-e^{-\beta/2}},
\end{equation}
\begin{equation}\label{eq29}
\beta=
    \begin{cases}
        \frac{\eta-\hat{R}^{\tau}_{max}}{\eta}&\text{ if }\hat{R}^{\tau}_{max}\leq \eta,\\
        \frac{\hat{R}^{\tau}_{max}-\eta}{\Delta R-\eta}&\text{ otherwise}.
    \end{cases}
\end{equation}

In Equation~\ref{eq27}, the calibration term (i.e., the last term) is composed of two factors: the decay factor $e^{-\epsilon}$ and the moving factor $\xi$. The decay factor controls the effect of the calibration term by decreasing its influence as $\epsilon$ grows. The moving factor $\xi$ depends on the gap between the perturbed radius and the calibration center. As shown in Equation~\ref{eq28}, the moving factor will be negative when the perturbed radius $\hat{R}^{\tau}_{max}$ is larger than $\eta$, and positive otherwise. The gap between $\hat{R}^{\tau}_{max}$ and $\eta$ is multiplied by a factor in the range of $[0,1]$ to avoid obtaining an extreme value or a constant value. Therefore, $\beta$ is calculated using Equation~\ref{eq29}, which normalizes the gap between $\hat{R}^{\tau}_{max}$ and $\eta$ to $[0,1]$. Then $\beta /2$ is fed into an activation function $sigmoid(\cdot)$, which provides the non-linear property and prevents the moving step size from having an exact value of $|\eta-\hat{R}^{\tau}_{max}|$. After the calibration, we can determine the perturbed point set of the input trajectory (Line~\ref{alg4line5} of Algorithm~\ref{alg4}).

\subsection{Analysis of Privacy and Utility}
In this subsection, we prove that the proposed ATP mechanism satisfies $\epsilon$-LDP.

\begin{theorem}\label{theorem4}
    The ATP mechanism satisfies $\epsilon$-LDP.
\end{theorem}
\begin{proof}
    The ATP mechanism restricts the spatial regions of the two trajectory copies in a sequential manner. The perturbed anchors $\hat{p}_{\alpha}^{\prime}$ and $\hat{p}_{\alpha}^{\ast}$ are obtained using EM, each satisfying $\epsilon_{1}/2$-LDP. Then $\hat{R}^{\prime}$ and $\hat{R}^{\ast}$ are calculated to restrict the perturbation domains $\hat{\mathcal{P}}_{\alpha}^{\prime}$ and $\hat{\mathcal{P}}_{\alpha}^{\ast}$ of the points in $\hat{\tau}^\prime$ and $\hat{\tau}^\ast$, respectively. The anchors are perturbed, and then the largest distance between the perturbed $\hat{p}_{\alpha}^{\prime}$ (or $\hat{p}_{\alpha}^{\ast}$) and the points in $\tau^\prime$ and $\tau^\ast$ is calculated as $\hat{R}^{\prime}$ (or $\hat{R}^{\ast}$) by the SW mechanism, which satisfies $\epsilon_{2}/2$-LDP. The calibration operation does not access any raw data because the test values are independent of the trajectory and the distances used to identify $\mathcal{P}^{R}$ are based on the perturbed $\hat{p}_{\alpha}^{\prime}$ (or $\hat{p}_{\alpha}^{\ast}$). By the post-processing property, the calibration operation does not require additional privacy budget. Next, $\mathcal{P}$ is replaced once with $\hat{\mathcal{P}}_{\alpha}^{\prime}$ and then with $\hat{\mathcal{P}}_{\alpha}^{\ast}$. The TP mechanism that satisfies $\epsilon_{3}$/2-LDP is then invoked. Again, by the post-processing theorem, the replacements do not consume additional privacy budget. Finally, the sequential composition theorem ensures that the ATP mechanism satisfies $\epsilon$-LDP, where $\epsilon=\epsilon_{1}+\epsilon_{2}+\epsilon_{3}$.
\end{proof}

To understand the impact of different $R$ values, we perform a utility analysis by estimating the upper bound of the probability that the anchor points $p_{\alpha}$ in the trajectory are not covered by the circle whose radius is $R$ and origin is the perturbed $\hat{p}_{\alpha}$.

\begin{theorem}\label{theorem3}
    Let $\mathcal{C}_{\tau}(p,R)=\{p^{\prime}\in \mathcal{P}|dist(p,p^{\prime})\leq R\}$ and $u_{x}$ be the minimum score that satisfies $u(p_{\alpha},p)=-dist(p_{\alpha},p)> -R$. If $R=2t\Delta u/\epsilon-u_x$, we have
    \begin{equation}\notag
        Pr[p_{\alpha}\notin \mathcal{C}_{\tau}(\hat{p}_{\alpha},R)]\leq \frac{|\mathcal{P}-\mathcal{C}_{\tau}(p_{\alpha},R)|}{|\mathcal{C}_{\tau}(p_{\alpha},R)|}exp(-t).
    \end{equation}
\end{theorem}

\begin{proof}\allowdisplaybreaks[2]
	We have the probability 
    \begin{align}
        \notag&Pr[p_{\alpha}\notin \mathcal{C}_{\tau}(\hat{p}_{\alpha},R)]=Pr[\hat{p}_{\alpha}\notin \mathcal{C}_{\tau}(p_{\alpha},R)] = Pr[u(p_{\alpha},\hat{p}_{\alpha})\leq -R]\\
        \notag&=\frac{\sum_{p_{c}\in \mathcal{P}-\mathcal{C}_{\tau}(p_{\alpha},R)}exp(\frac{\epsilon u(p_{\alpha},p_{c})}{2\Delta u})}{\sum_{p\in \mathcal{P}}exp(\frac{\epsilon u(p_{\alpha},p)}{2\Delta u})} \leq \frac{|\mathcal{P}-\mathcal{C}_{\tau}(p_{\alpha},R)|exp(\frac{\epsilon (-R)}{2\Delta u})}{|\mathcal{C}_{\tau}(p_{\alpha},R)|exp(\frac{\epsilon u_{x}}{2\Delta u})}\\
        \notag&= \frac{|\mathcal{P}-\mathcal{C}_{\tau}(p_{\alpha},R)|}{|\mathcal{C}_{\tau}(p_{\alpha},R)|}exp(-\frac{\epsilon(R+u_x)}{2\Delta u}).
    \end{align}
	Let $R=\frac{2t\Delta u}{\epsilon}-u_x$, then we have
	\begin{align}\notag
		Pr[p_{\alpha}\notin \mathcal{C}_{\tau}(\hat{p}_{\alpha},R)]\leq 
		\frac{|\mathcal{P}-\mathcal{C}_{\tau}(p_{\alpha},R)|}{|\mathcal{C}_{\tau}(p_{\alpha},R)|}exp(-t),
	\end{align}
	which completes the proof.
\end{proof}

According to Theorem~\ref{theorem3}, for a certain $R$, the probability of an anchor point $p_{\alpha}$ being not covered by the circle with the origin at $\hat{p}_{\alpha}$ and radius $R$ decreases exponentially as the radius $R$ increases. For the sake of data utility, a small probability can reduce the error of anchor point selection. Furthermore, a small $R$ can avoid resulting in an overly large perturbation domain and introducing excessive noise. Hence the selection of $R$ becomes a trade-off with respect to the privacy budget $\epsilon$. Besides, this selection of $R$ is also related to the point distribution of a dataset, i.e., the minimum score $u_x$ that satisfies $u(p_{\alpha},p)=-dist(p_{\alpha},p)> -R$. A smaller $u_x$ results in a tighter bound in the theorem. For this reason, we propose an adaptive approach, as described in Sections 5.1 and 5.2, to find a reasonable $R$ to improve data utility.

%% file: 6_experiments_conclusion.tex
\section{Experiments}\label{experiment}
\subsection{Experimental Setting}
\subsubsection{Datasets}
In the experiments, we use three real-world and one synthetic datasets, namely, NYC, CHI, CLE, and CPS. NYC consists of check-in trajectories in New York City extracted from the Foursquare dataset~\cite{YZZ2014}, while CHI and CLE, extracted from the Gowalla dataset~\cite{CML2011}, consist of check-in trajectories in Chicago and Cleveland, respectively.\footnote{We select the points in the approximate range of [87.4W—88W, 41.6N—42N] for Chicago and [122.46W—122.9W, 45.4N—45.6N] for Cleveland.} 
We consider the 1,000 most popular points as $\mathcal{P}$ to generate CHI and CLE, and the 2,000 most popular POIs in NYC. 
\textcolor{black}{For a fair comparison, we adopt the same preprocessing steps used in the previous study~\cite{CCF2021}. We randomly delete the points that appear within a 10-minute duration in each trajectory until only one point remains.} If the time interval between any two adjacent points in a trajectory exceeds three hours, we split it into two trajectories. 
After these preprocessing steps, we obtain 7,951, 3,162, and 2,794 trajectories in NYC, CHI, and CLE, respectively. For CPS, \textcolor{black}{we follow the previous study~\cite{CCF2021} to generate trajectories} on the campus of the University of British Columbia\footnote{https://github.com/UBCGeodata/ubcv-buildings}. We take 262 campus buildings as $\mathcal{P}$ and generate 4,000 trajectories.

\subsubsection{Baselines and Parameter Setting}\label{para}
The only study that satisfies pure $\epsilon$-LDP is NGRAM mechanism~\cite{CCF2021}, which perturbs POI trajectories by incorporating external knowledge. As aforementioned, it is often difficult to acquire such external knowledge in practice, which is the key motivation of our paper. Hence, we consider NGRAM mechanism without any additional knowledge as a baseline and set the grid granularity to $3$ or $4$ for different datasets. Another baseline is 
a direct application of the exponential mechanism (referred to as EXP). 
It perturbs each point in a trajectory by using the same utility function used in the mechanisms proposed in this study, i.e., $-dist(\cdot)$. The last baseline is CGM~\cite{BYX2021}, a state-of-the-art mechanism for streaming data collection under $(\epsilon,\delta)$-LDP. We normalize the latitude and longitude of each point in the same way as in the previous study~\cite{BYX2021}, by setting $\delta=10^{-2}$ or $10^{-1}$, and $C=0.1$. For all the mechanisms, we use the Haversine distance as the distance metric. 

As for the privacy budget allocation scheme in the ATP mechanism, $\epsilon^\prime=\epsilon^\ast=\frac{\epsilon}{2}$ are used to perturb $\tau^{\prime}$ and $\tau^{\ast}$ respectively. For $\epsilon^\prime$, $\frac{\epsilon^{\prime}}{4}$ is used to determine the region. As the region size plays a more significant role in determining the trajectory region, a quarter of $\frac{\epsilon^\prime}{4}$ is used to perturb the trajectory anchor while the other three quarters are used to perturb the radius. The remaining budget (i.e., $\frac{3\epsilon^\prime}{4}$) is used to perturb the directions and points in $\tau^\prime$. As the directions have a larger impact on the perturbations of trajectories, three quarters of $\frac{3\epsilon^\prime}{4}$ are uniformly divided to perturb the directions while another quarter is uniformly divided to perturb the points. The allocation of $\epsilon^\ast$ for perturbing $\tau^\ast$ is the same as $\tau^\prime$. As for the TP mechanism, we use the same budget allocation strategy for the perturbations of directions and points in the ATP mechanism. All the mechanisms are executed $5$ times and the average is plotted.


\subsection{Results}

\subsubsection{Measures}
We evaluate the utility of the mechanisms using two measures adopted in the previous study~\cite{CCF2021}. The first measure is the mean normalized error (NE), which is the normalized distance between each point of a perturbed trajectory and the corresponding point of the original trajectory:
\begin{equation}\notag
    NE=\frac{1}{|\mathcal{T}|}\sum\nolimits_{i=1}^{|\mathcal{T}|}\frac{1}{|\tau_{i}|}\sum\nolimits_{j=1}^{|\tau_{i}|}dist(\hat{p}_{j}, p_{j}),
\end{equation}
where $|\mathcal{T}|$ is the number of trajectories, and $dist(\cdot)$ denotes the Haversine distance. \textcolor{black}{We normalize the results by the maximum distance in a dataset to make the results easier to understand.} The other measure is the preservation range query (PRQ), which evaluates whether each point of the perturbed trajectory is within the $\delta$ (km) range of the corresponding true point:
\begin{equation}\notag
    PRQ=\frac{1}{|\mathcal{T}|}\sum\nolimits_{i=1}^{|\mathcal{T}|}\frac{1}{|\tau_{i}|}\sum\nolimits_{j=1}^{|\tau_{i}|}\pi
(p_{j},\hat{p}_{j},\delta)\times 100\%,
\end{equation}
\begin{equation}\notag
\pi
(p_{j},\hat{p}_{j},\delta)=
    \begin{cases}
        1&\text{ if } \ dist(\hat{p}_j,p_j)\leq \delta,\\
        0&\text{ otherwise},
    \end{cases}
\end{equation}
where $dist(\cdot)$ is also measured by the Haversine distance. A larger PRQ value indicates better performance.

\begin{figure}[t]
    \centering
    \includegraphics[width=0.95\columnwidth]{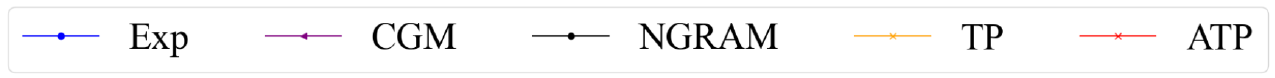} \\
    \subfloat[NYC]{
    \includegraphics[width=0.48\columnwidth]{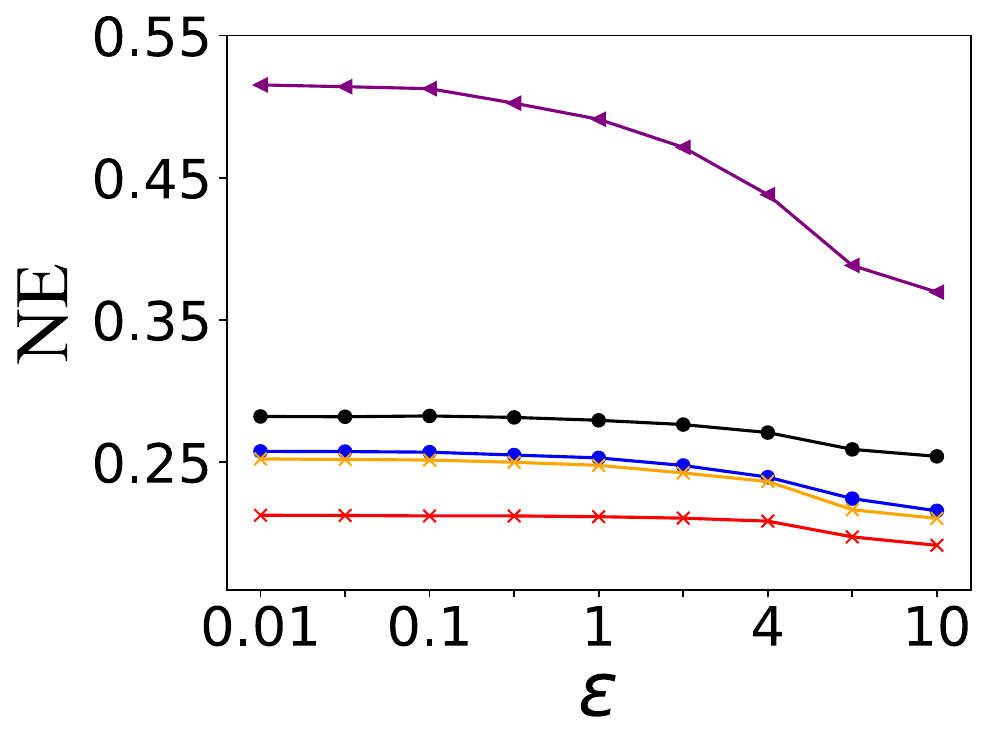}
    }\hfill
    \subfloat[CHI]{
    \includegraphics[width=0.48\columnwidth]{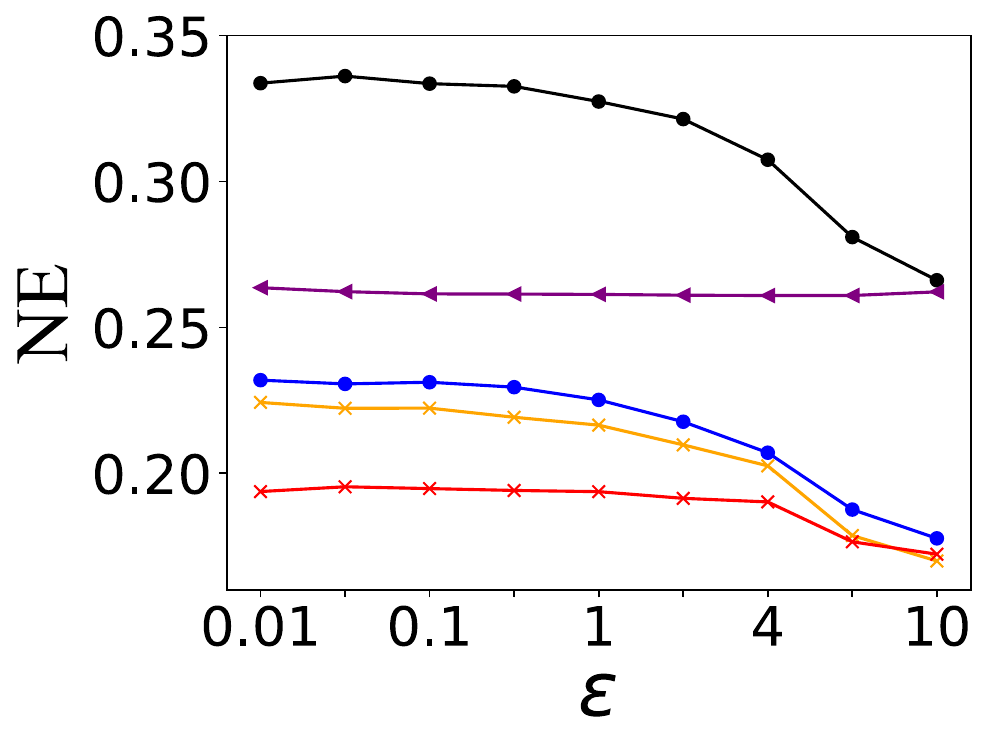}
    }\hfill
    \subfloat[CLE]{
    \includegraphics[width=0.48\columnwidth]{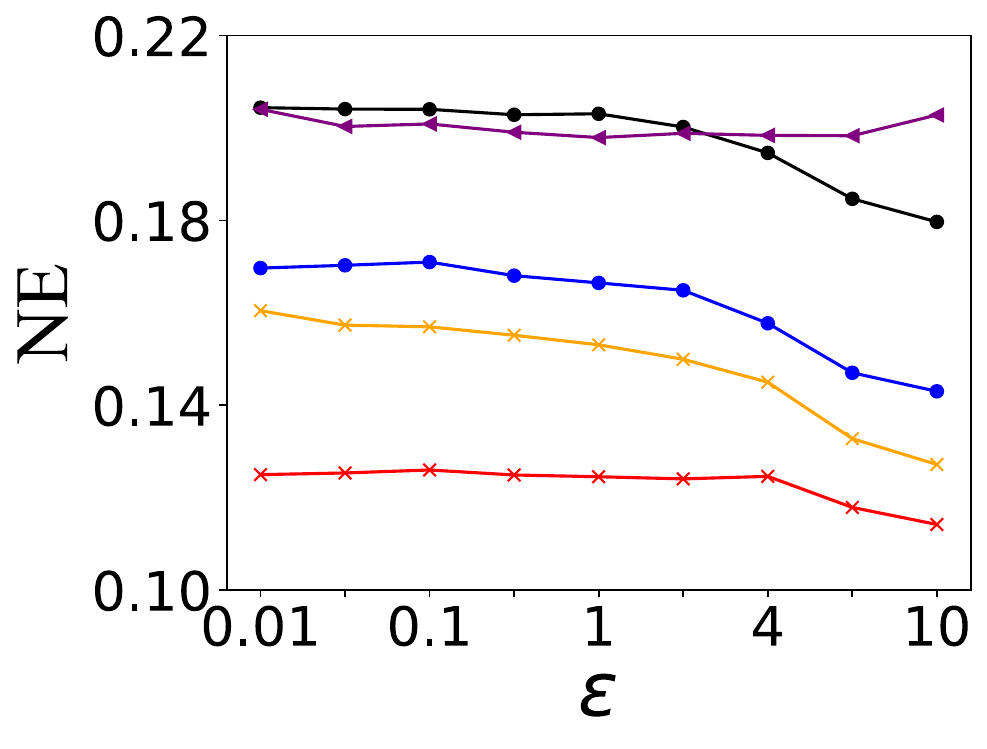}
    }\hfill
    \subfloat[CPS]{
    \includegraphics[width=0.48\columnwidth]{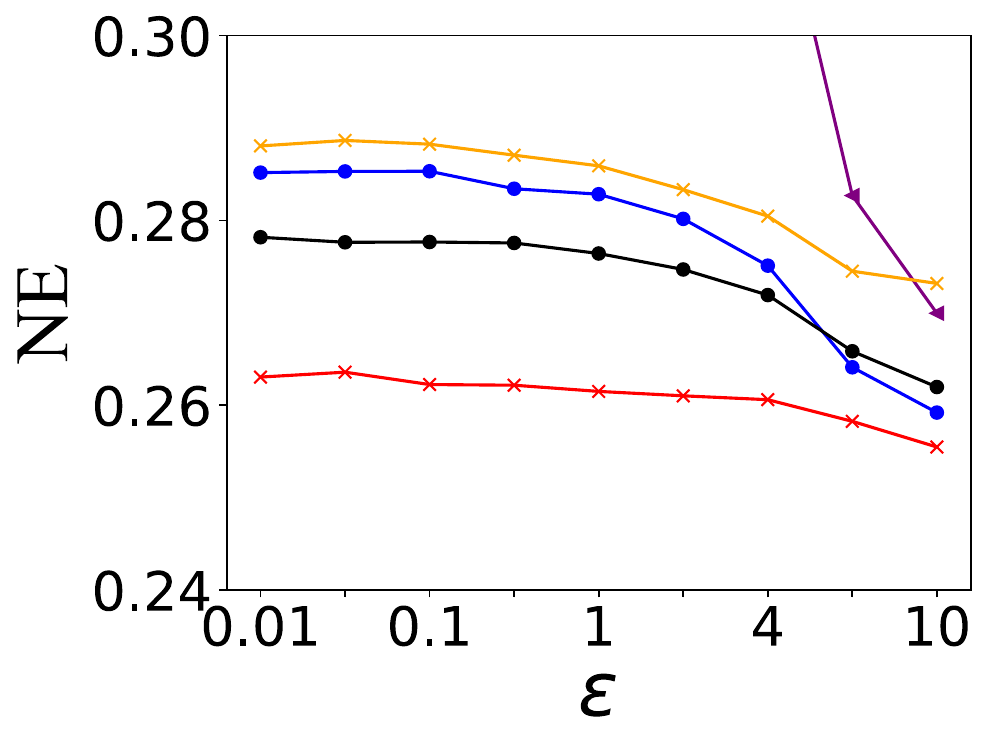}
    }
    \caption{Mean normalized errors (NEs) for different methods when varying privacy budget $\epsilon$.}
    \label{fig1}
\end{figure}

\subsubsection{Evaluation of Utility}
We first compare our proposed mechanisms with the baseline mechanisms in terms of NE. The results are presented in Figure~\ref{fig1}. The proposed mechanisms outperform the baselines in most situations. From Figure~\ref{fig1}, we observe that the TP mechanism (i.e., the orange line) outperforms the baselines on all real-world datasets. Although the EM used to sample points in the TP mechanism has difficulty in choosing points near the original point when $\epsilon$ is small, the use of the direction constraint can effectively enhance the utility. In NGRAM, since external knowledge is not used, the hierarchical decomposition in this mechanism cannot achieve significant performance improvement, and the large bi-gram universe adversely affects the utility. In CGM, the unbounded noise largely impacts the resulting utility. Although the TP mechanism performs worse than the baselines on the synthetic dataset, the ATP mechanism (represented by the red line) performs significantly better than other mechanisms, showing that the trajectory region constraint is more helpful when the points in a dataset are distributed more uniformly, such as the CPS dataset.
The better performance is partially due to the fact that the SW mechanism used to perturb the region size $R$ is more accurate when $\epsilon$ is large. The proposed adaptive calibration method is also advantageous to eliminate extreme values of region sizes due to the perturbation of the SW mechanism. As a result, even when $\epsilon$ is small, ATP can restrict trajectory regions better and achieve better performance.

With regard to the PRQ, the ATP mechanism outperforms the baseline mechanisms on all three real-world datasets as shown in Figure~\ref{fig2}. We set $\delta$ to 1, 2, or 4, and observe the changes in $PRQ$ when $\epsilon$ varies from 1 to 10 on different datasets (we set $\delta$ to 0.25, 0.5, or 1 for the CPS dataset due to the small entire area). The value of $PRQ$ increases as the privacy budget increases. For the CPS dataset, although the ATP mechanism performs worse than EXP when $\delta=0.25$, it performs better when $\delta$ becomes larger, especially when $\epsilon$ is small. 

\begin{figure}[t]
    \centering
    \includegraphics[width=0.95\columnwidth]{NE/legend_new.pdf} \\
    \subfloat[NYC $\delta$=1]{
    	\includegraphics[width=0.29\columnwidth]{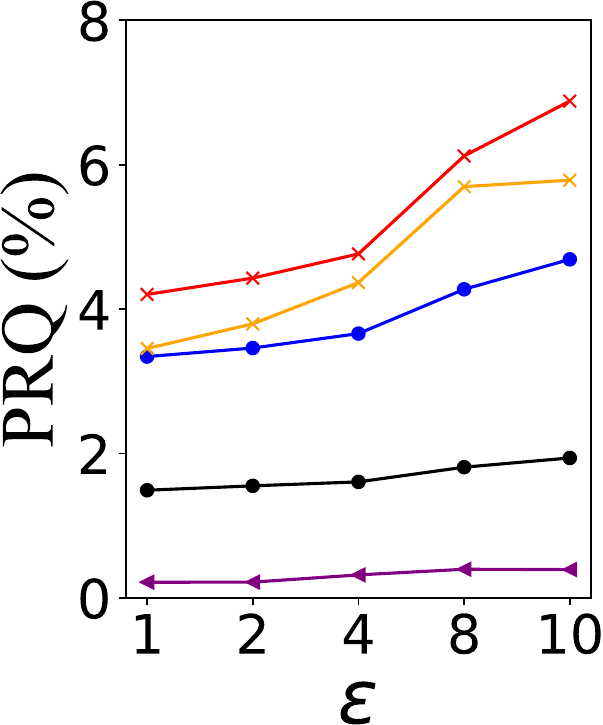}
    }\hfill
	\subfloat[NYC $\delta$=2]{
		\includegraphics[width=0.31\columnwidth]{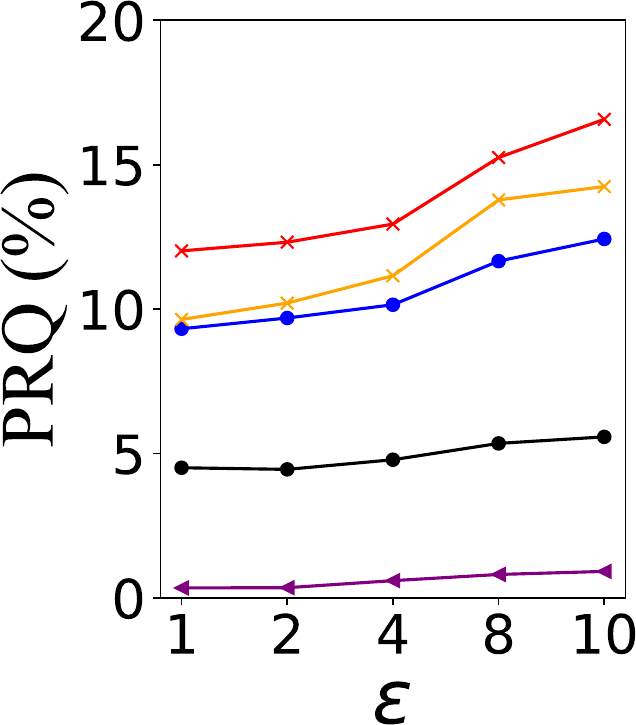}
	}\hfill
	\subfloat[NYC $\delta$=4]{
		\includegraphics[width=0.31\columnwidth]{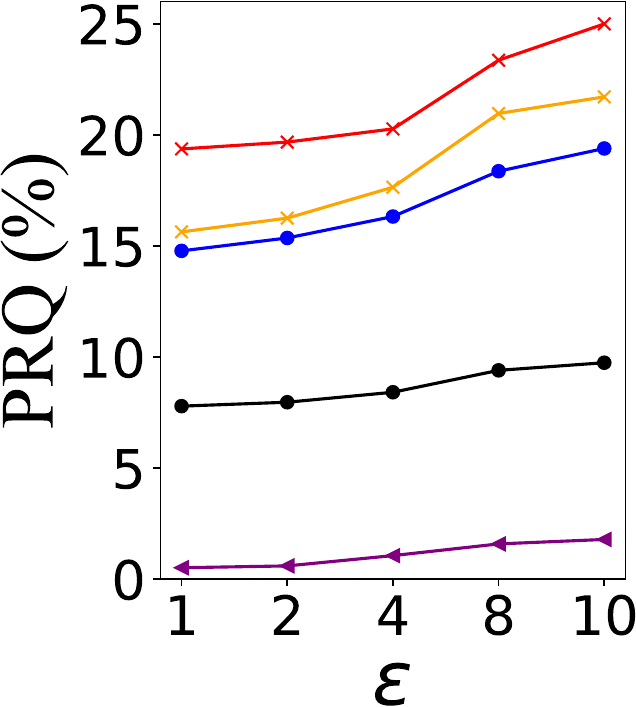}
	}\hfill

    \subfloat[CHI $\delta$=1]{
    	\includegraphics[width=0.3\columnwidth]
    {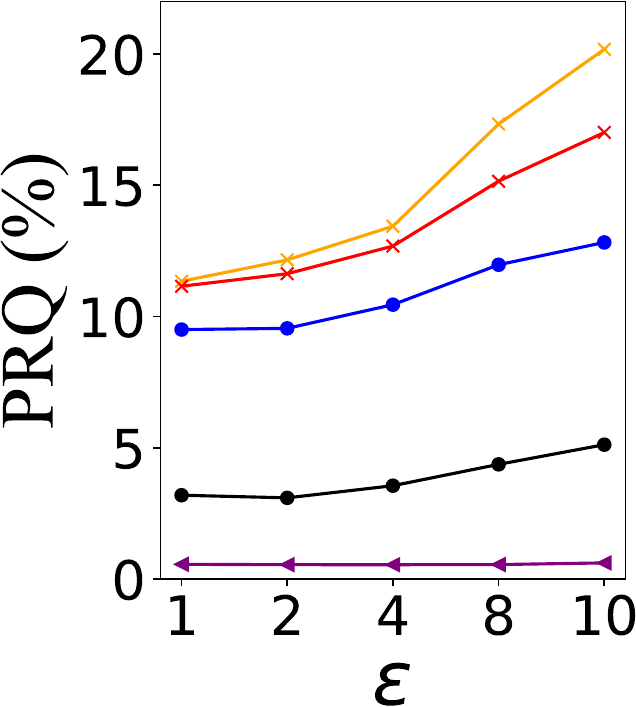}
    }\hfill
	\subfloat[CHI $\delta$=2]{
		\includegraphics[width=0.3\columnwidth]
		{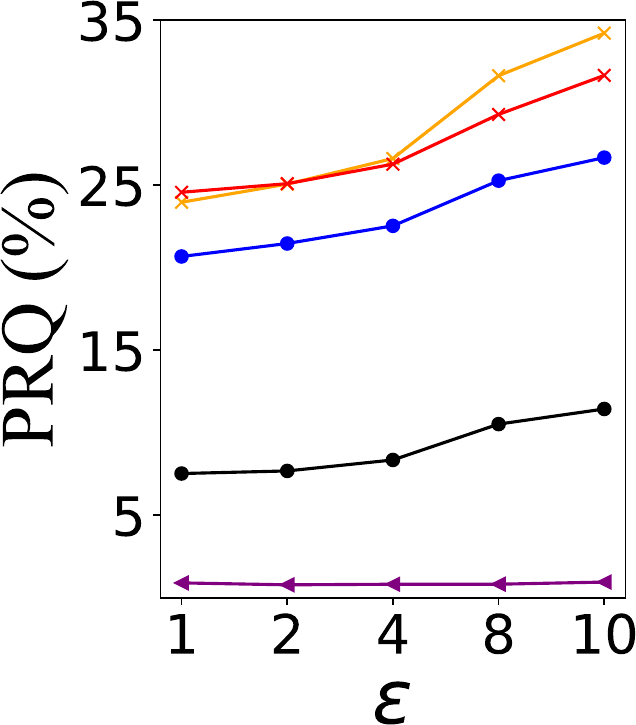}
	}\hfill
	\subfloat[CHI $\delta$=4]{
		\includegraphics[width=0.3\columnwidth]
		{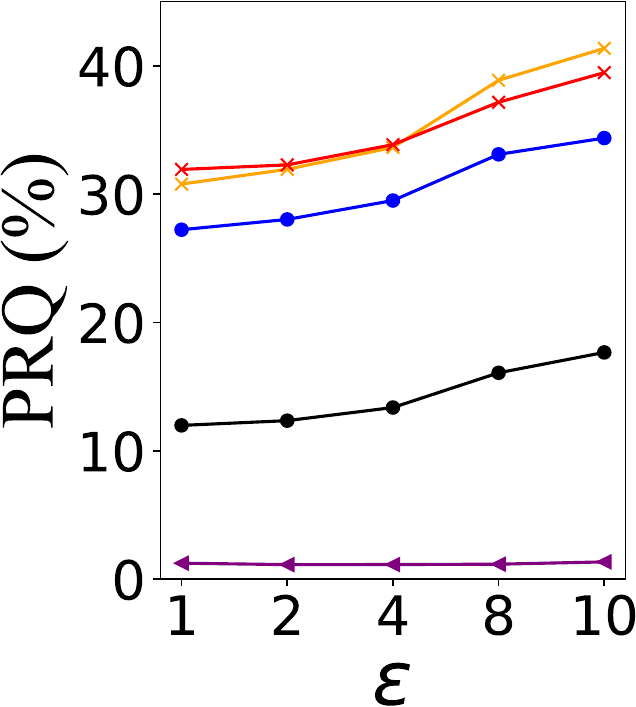}
	}\hfill

    \subfloat[CLE $\delta$=1]{
    	\includegraphics[width=0.3\columnwidth]
    {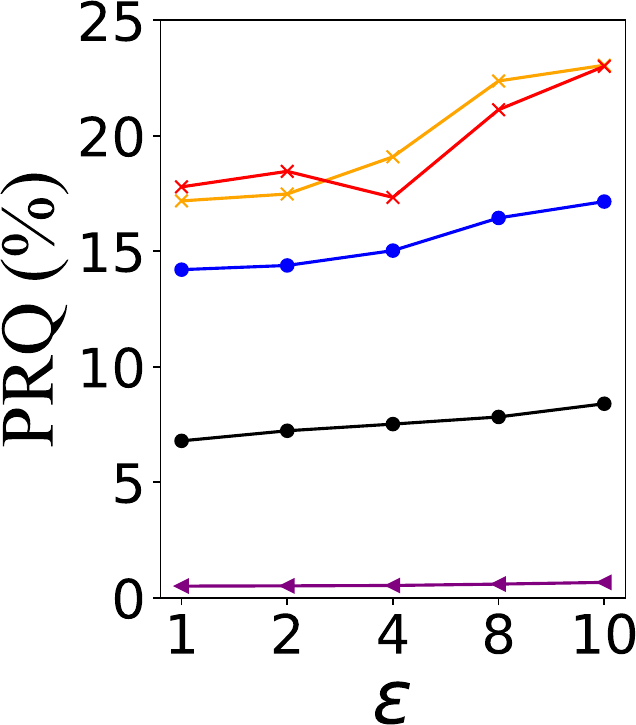}
    }\hfill
	\subfloat[CLE $\delta$=2]{
		\includegraphics[width=0.3\columnwidth]
		{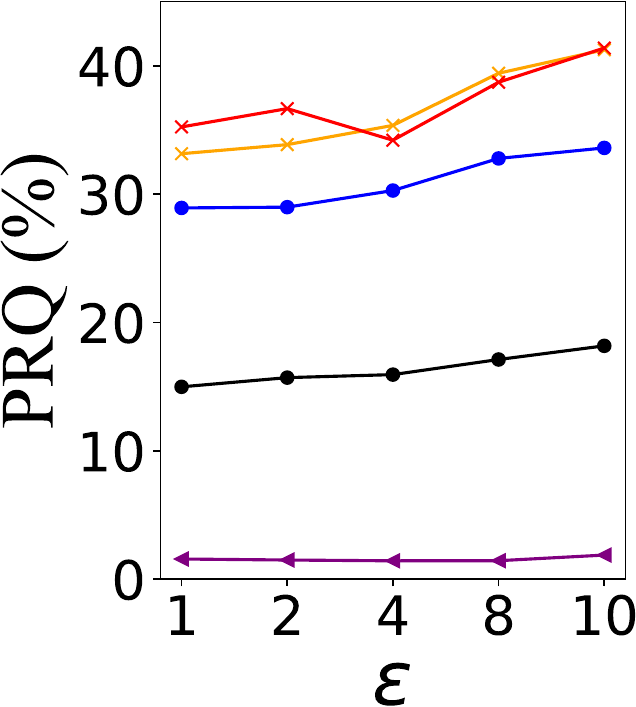}
	}\hfill
	\subfloat[CLE $\delta$=4]{
		\includegraphics[width=0.3\columnwidth]
		{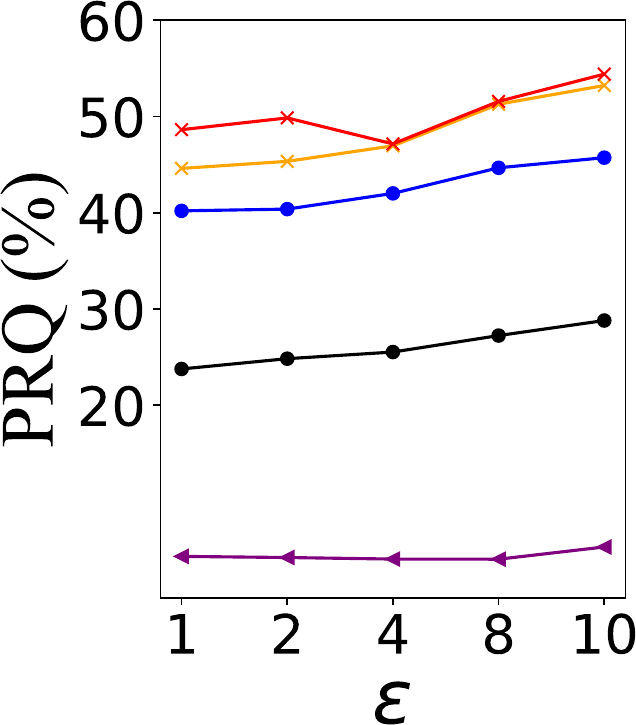}
	}\hfill

    \subfloat[CPS $\delta$=0.25]{
    	\includegraphics[width=0.32\columnwidth]
    {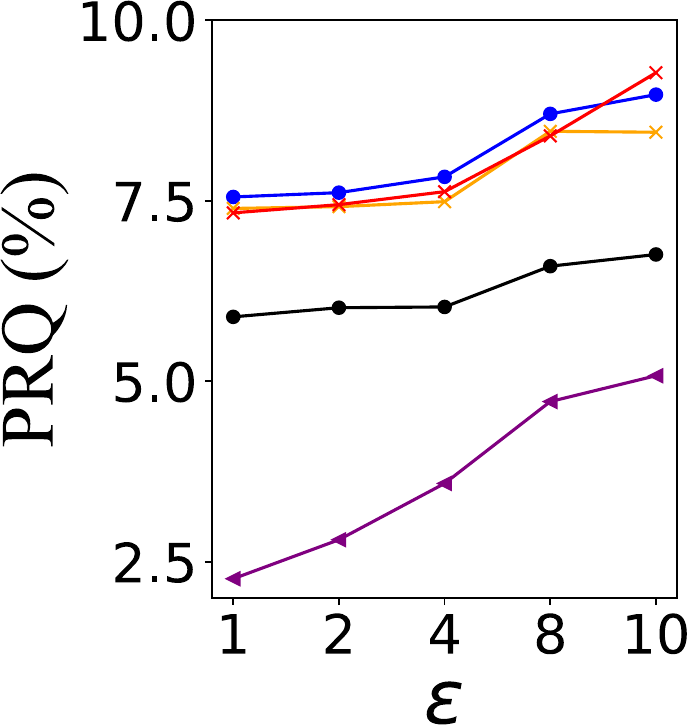}
    }\hfill
    \subfloat[CPS $\delta$=0.5]{
    	\includegraphics[width=0.3\columnwidth]
    {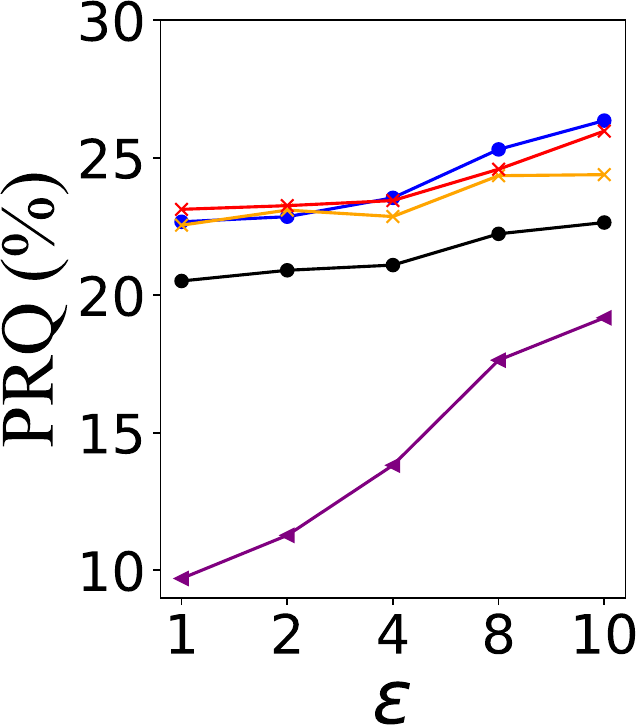}
    }\hfill
    \subfloat[CPS $\delta$=1]{
    	\includegraphics[width=0.3\columnwidth]
    {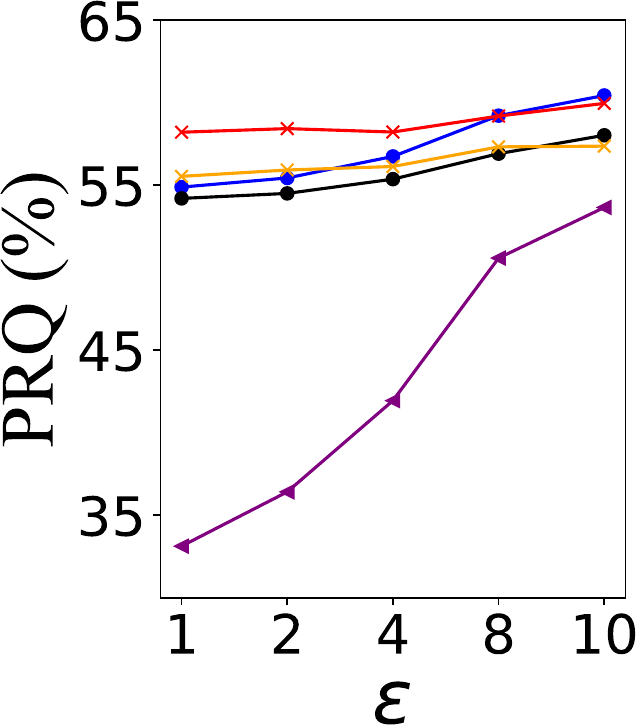}
    }
    \caption{Preservation range queries (PRQs) vs. privacy budget $\epsilon$ under different ranges $\delta$.}
    \label{fig2}
\end{figure}

\begin{figure}[t]
    \centering
    \subfloat[NYC]{\includegraphics[width=0.31\columnwidth]{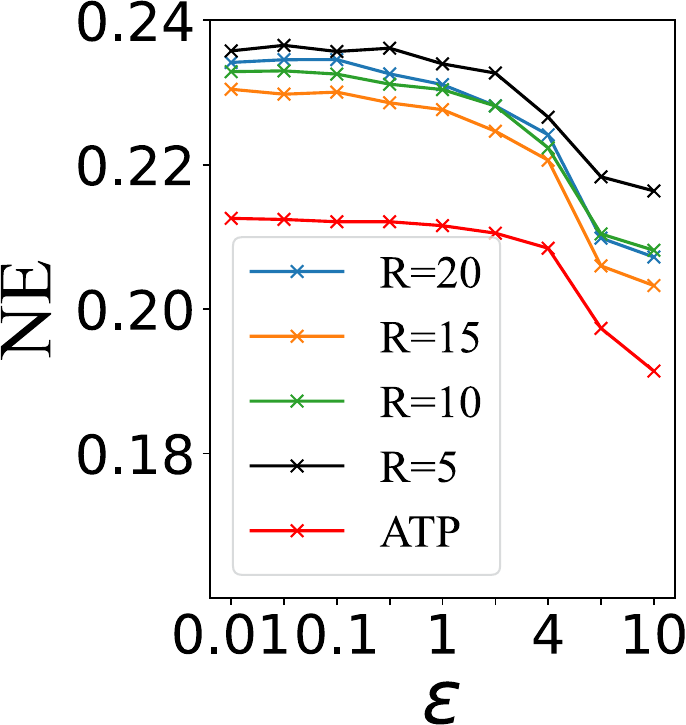}
    }\hfill
    \subfloat[NYC]{\includegraphics[width=0.31\columnwidth]{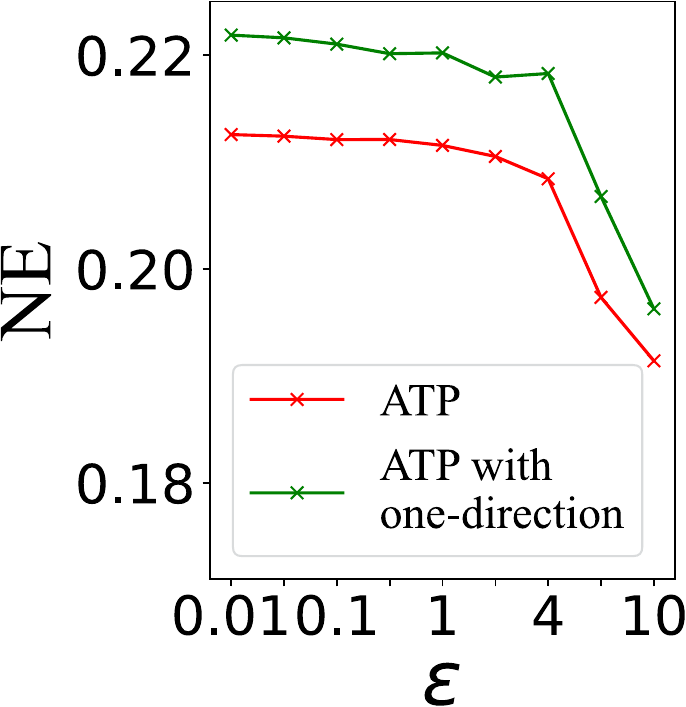}
    }\hfill
   \subfloat[NYC]{\includegraphics[width=0.31\columnwidth]{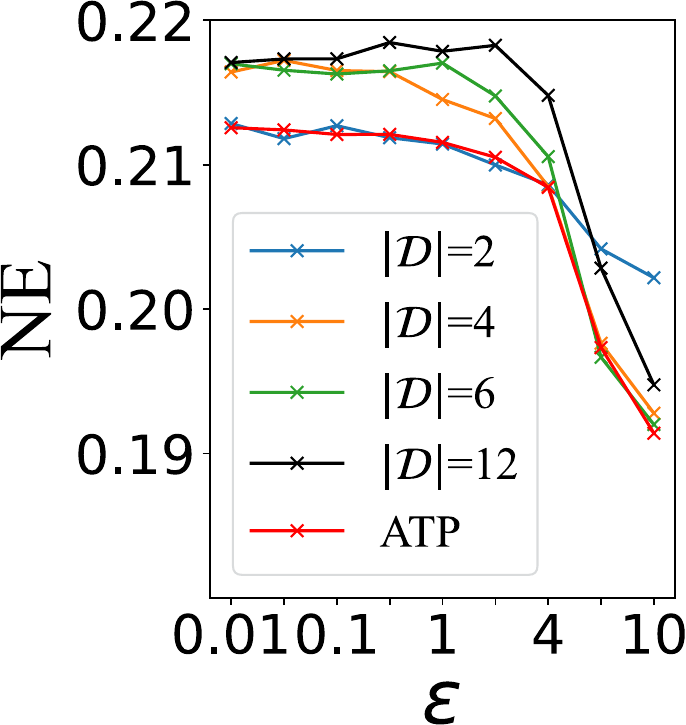}
   }\hfill
   \subfloat[CLE]{\includegraphics[width=0.31\columnwidth]
    {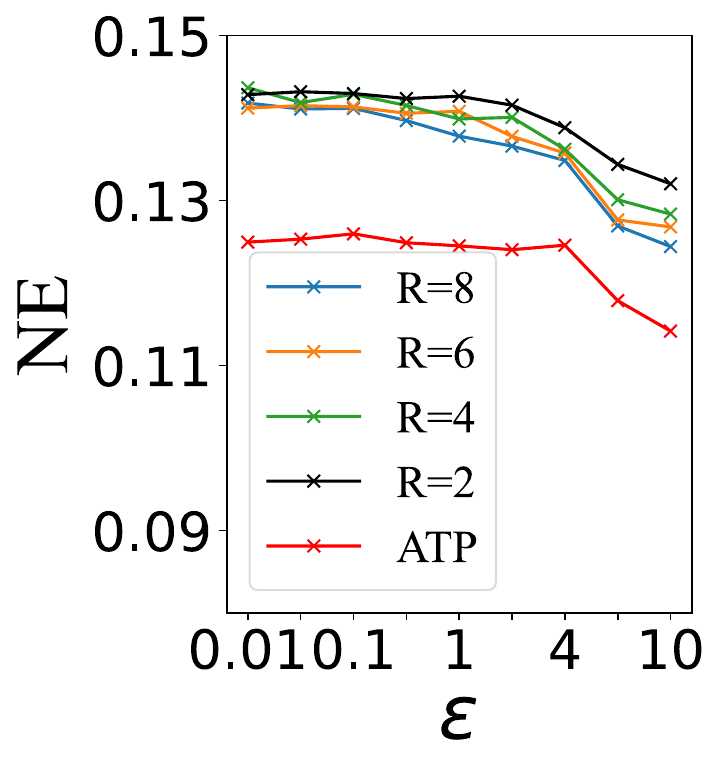}
    }\hfill
    \subfloat[CLE]{\includegraphics[width=0.31\columnwidth]
    {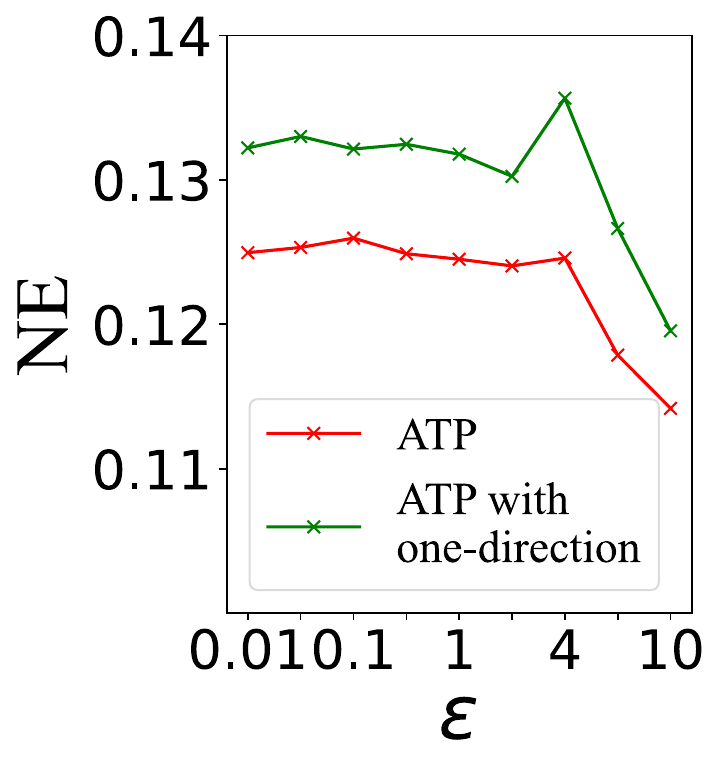}
    }\hfill
    \subfloat[CLE]{\includegraphics[width=0.31\columnwidth]
    {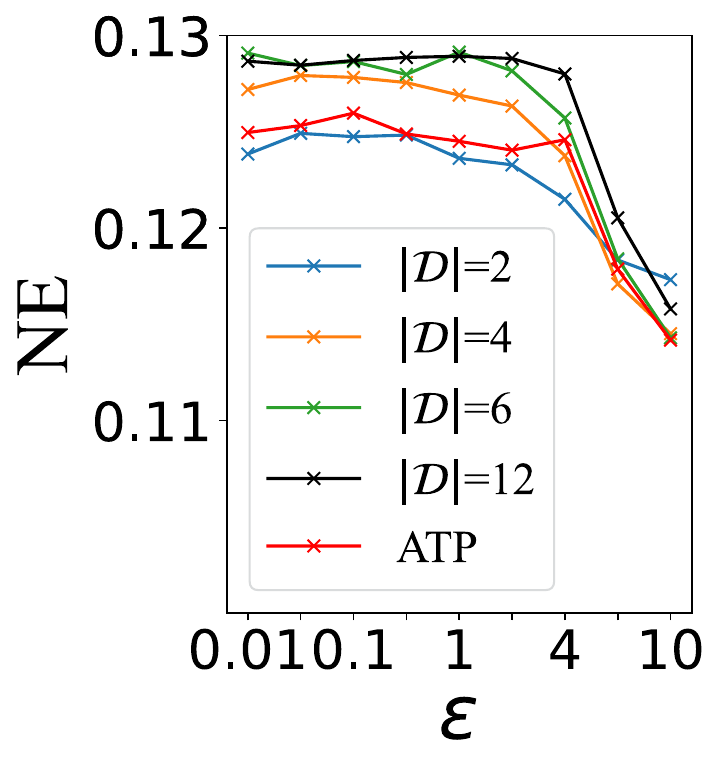}}
    \caption{Mean normalized errors (NEs) of different radius $R$, uni-directional information and direction granularities under varying privacy budget $\epsilon$. 
    }
    \label{fig4}
\end{figure}

\begin{figure}[t]
    \centering
    \subfloat[Size of $\mathcal{P}$]{
    \includegraphics[width=0.31\columnwidth]
    {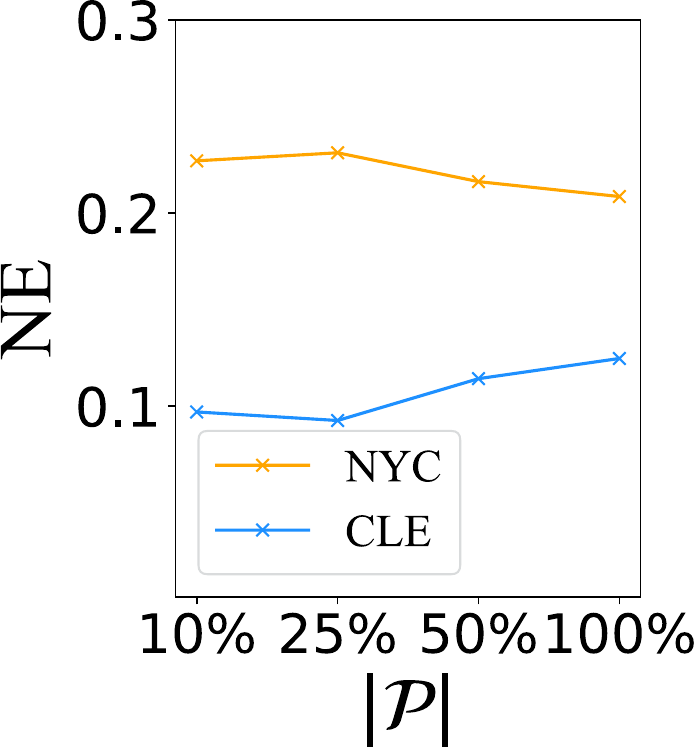}
    }\hfill
    \subfloat[Trajectory Length]{
    \includegraphics[width=0.32\columnwidth]
    {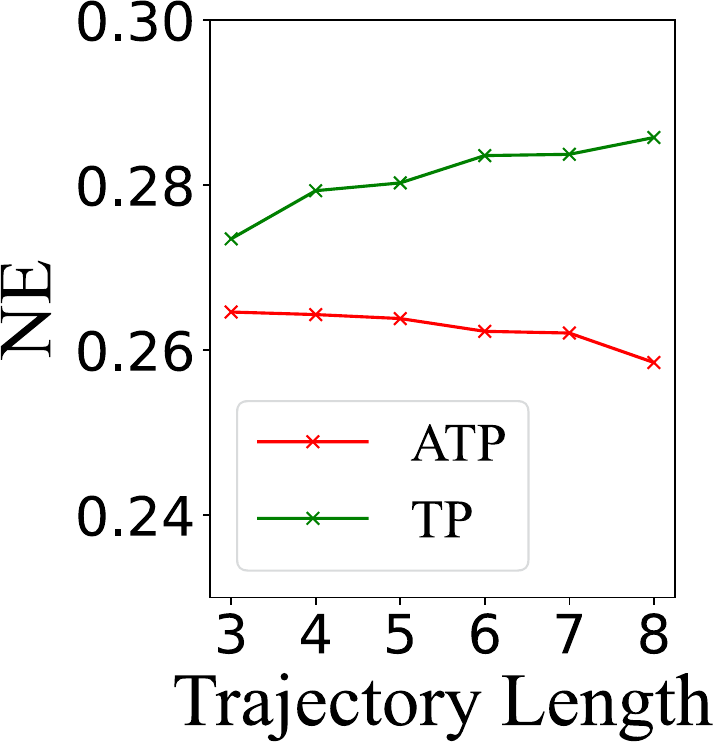}
    }\hfill
    \subfloat[Reach Bound]{
    \includegraphics[width=0.31\columnwidth]
    {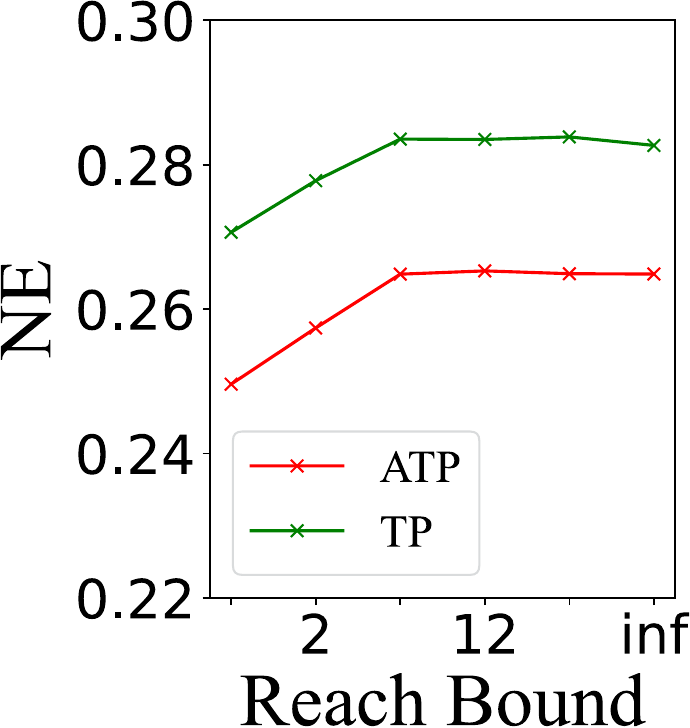}
    }
    \caption{Mean normalized errors (NEs) under different sizes of $\mathcal{P}$, trajectory length and reach bound ($\epsilon=4$). }
    \label{fig5}
\end{figure}

\subsubsection{Effects of Parameters}
Due to the space limitation, we only report the results on the NYC and CLE datasets. Similar results and conclusions can be observed and drawn on the CHI and CPS datasets. First, we study the performance of the ATP mechanism with varying fixed $R$ values. We experiment with different $R$ values on different datasets: $R\in\{5, 10, 15, 20\}$ for NYC and $R\in\{2, 4, 6, 8\}$ for CLE. As shown in Figures~\ref{fig4}a and \ref{fig4}d, 
as $\epsilon$ increases, a larger $R$ value can cover more points in the trajectories whose original regions are large, leading to better performance. We observe that the value of $R$ that results in the best performance is not always the largest or smallest on different datasets, irrespective of whether $\epsilon$ is small or large. Therefore, we advise to use our adaptive method to determine trajectory regions, which can restrict trajectory regions more precisely and thus obtain the best performance.

Next, we investigate whether the bi-directional information improves the performance of our mechanism. From Figures~\ref{fig4}b and \ref{fig4}e, it can be seen that the ATP mechanism using bi-directional information outperforms the same mechanism using uni-directional information, which confirms that additional direction information enables more accurate restriction of trajectory regions.

We also study the effects of different direction granularities $|\mathcal{D}|\in \{2,4,6,12\}$ on the performance. As shown in Figures~\ref{fig4}c and \ref{fig4}f, when $\epsilon$ is small, the ATP mechanism with coarse-grained directions performs better than that with fine-grained directions. 
\textcolor{black}{The choice of the cardinality of discrete directions mainly depends on the given privacy budget. Note that direction selection aims to restrict the perturbation domain, i.e.,  $|\mathcal{D}|=g$ means to narrow down the domain to $\frac{1}{g}$, which reduces the location perturbation noise and thus improves the data utility. However, direction perturbation itself also introduces some noise, where finer direction granularity comes with heavier perturbation under LDP. 
Therefore, for a small $\epsilon$, we tend to select a small $|\mathcal{D}|$ value to avoid large direction perturbation noise. On the contrary, as $\epsilon$ increases, direction perturbation becomes less eminent, and, therefore, we can select a larger $|\mathcal{D}|$ value to enhance the utility more significantly. 
}


\begin{table}[t]
\centering
\caption{Comparison of success probability. Each value is calculated as per Equation~\ref{eqprob}. The boldfaced values are the largest values under different $\epsilon$ values.}

\resizebox{0.85\columnwidth}{!}{\begin{tabular}{c|c|c|c|c}
\hline
     & $|\mathcal{D}|$=2 & $|\mathcal{D}|$=4 & $|\mathcal{D}|$=6 & $|\mathcal{D}|$=12 \\ \hline
$\epsilon$=0.01 & \textbf{0.25035156}         & 0.20871446                 & 0.16699901                 & 0.09869639                 \\ \hline
$\epsilon$=0.05 & \textbf{0.25175778}        & 0.21024432                 & 0.16833461                 & 0.09954752                  \\ \hline
$\epsilon$=0.1  & \textbf{0.25351539}        & 0.21216864                 & 0.17001818                 & 0.10062269                   \\ \hline
$\epsilon$=0.5  & \textbf{0.26754921}        & 0.22803653                 & 0.18405465                 & 0.10968742                  \\ \hline
$\epsilon$=1    & \textbf{0.28492633}        & 0.24901168                 & 0.20304037                 & 0.12224172                  \\ \hline
$\epsilon$=2    & \textbf{0.3185154}        & 0.29434453                  & 0.24584151                 & 0.15185139                  \\ \hline
$\epsilon$=4    & 0.37745749                 & \textbf{0.39365306}        & 0.34902402                 & 0.23167506                  \\ \hline
$\epsilon$=8    & 0.45232527                 & 0.57649644                 & \textbf{0.58164843}        & 0.47196792                 \\ \hline
$\epsilon$=10   & 0.47167379                 & 0.63974545                 & \textbf{0.6787087}        & 0.60876684                 \\ \hline
\end{tabular}}

\label{table1}
\end{table}

We then evaluate how the size of $\mathcal{P}$ impacts the utility of the ATP mechanism. We conduct experiments on the datasets with thousands of locations, i.e., NYC and CLE datasets with 1,000 and 2,000 locations, respectively, following the previous study~\cite{CCF2021}. We take the most popular 10\%, 25\%, 50\%, and 100\% locations from them to evaluate the performance of our mechanism with privacy budget $\epsilon=4$. We can observe in Figure~\ref{fig5}a that the impact of a large $|\mathcal{P}|$ is not too significant. This is because the anchor-based trajectory region constraint mitigates the negative influence of a large $|\mathcal{P}|$. 

Furthermore, we conduct experiments on the synthetic CPS dataset to assess the performance of our mechanisms under different tuning parameters. As with the previous study~\cite{CCF2021}, we generate 4,000 trajectories with different fixed lengths ranging from 3 to 8 and different reach bounds between each two points in a trajectory with different travel speeds from 1 to 16 km/hr, and with no reach constraint (i.e., $\infty$). As shown in Figures~\ref{fig5}b and ~\ref{fig5}c, overall, ATP outperforms TP, and their performance is quite stable with different trajectory lengths or different reach bounds. Besides, we have an interesting observation that the NE of ATP even slightly decreases with the increasing trajectory length, which is probably because more locations contribute to the effect of direction restriction. 


Finally, we validate our proposed direction granularity strategy. As aforementioned, we calculate the average probability of success to choose the correct discrete direction for different ranges and for each granularity $|\mathcal{D}|$ under different $\epsilon$ values. We set different radian intervals as $\Theta= \{\pi/2,\pi/4,\pi/6,\pi/12\}$, corresponding to different candidate granularities $|\mathcal{D}|\in \{2,4,6,12\}$. In Table~\ref{table1}, $\epsilon$ denotes the total privacy budget, instead of the budget that is used to compute the probability for direction perturbation (i.e., three quarters of $\frac{3\epsilon^{\prime}}{4}$ or $\frac{3\epsilon^\ast}{4}$, as described in Section~\ref{para}). \textcolor{black}{We can observe that the granularity chosen based on the average direction-preserving success probability almost always corresponds to one of the best granularities under different $\epsilon$ values in Figures~\ref{fig4}c and \ref{fig4}f.}

\begin{table}[t]
\centering
\caption{Comparison of average count differences (ACDs) for different methods ($\epsilon$=4).}

\resizebox{0.62\columnwidth}{!}{\begin{tabular}{c|c|c|c|c}
\hline
     & NYC & CHI & CLE & CPS \\ \hline
CGM & 17.7344         & 11.4397                 &12.1424        &73.446         \\ \hline
NGRAM & 13.5946         & 11.7533                 &10.2275        &32.3889         \\ \hline
TP & 10.1774        & 7.1965                 &5.8968             &25.8770     \\ \hline
ATP  & 10.9542        & 7.4568                 &6.1147            &28.4412       \\ \hline

\end{tabular}}
\label{table_ACD}
\end{table}

\subsubsection{Evaluation on Practical Applications}
\textcolor{black}{We further consider hotspot preservation as a practical application to demonstrate the advantages of our mechanisms over the existing works. Since we do not assume time information in our setting, we evaluate the average count difference (ACD), a popular metric for hotspot preservation~\cite{CCF2021}. Given a set of the most popular locations, ACD calculates the average absolute difference of their counts before and after perturbation. The experiments are conducted on four datasets with privacy budget $\epsilon=4$. In particular, for the NYC dataset, we evaluate the ACD of the top 50\% of popular locations, while for the CHI, CLE, and CPS datasets with less locations than NYC, we evaluate the ACD of their top 75\% of popular locations. As shown in Table~\ref{table_ACD}, we can observe that both TP and ATP mechanisms achieve significantly lower ACD than NGRAM and CGM on all datasets, which again confirms the superiority of our proposed methods. Without additional external knowledge, the unbounded noise will largely impact the performance of NGRAM and CGM. TP and ATP achieve similar performance. TP can obtain slightly lower ACD than ATP. We believe this is because ATP restricts the perturbation region of a trajectory, which may lead to a slightly more skewed hotspot count distribution. Nevertheless, ATP is still the best choice for a data collector to perform a wide range of analysis tasks.}

\section{Conclusion}\label{conclusion}
In this paper, we propose a novel private trajectory data collection mechanism called ATP that satisfies pure $\epsilon$-LDP. We model trajectory perturbation as a two-stage process, which first estimates discrete directions and then performs point perturbations. The proposed mechanism is further enhanced by adaptively restricting the trajectory region. To the best of our knowledge, our solution is the first to combine direction information with the perturbation of a trajectory under LDP. 
We also provide theoretical utility analysis of the region size and a guideline to choose a suitable direction granularity based on the given privacy budget. From the experimental results, we can observe an interesting correlation between the underlying distribution of the location points (e.g., their densities in different regions) and the direction granularity selection process. We believe exploring such a correlation is a valuable direction for future work. 
